\documentclass[aps,twocolumn,floatfix,superscriptaddress,nofootinbib]{revtex4-2}
\usepackage{booktabs}
\usepackage{graphicx}
\usepackage{braket}
\usepackage{stix}
\usepackage{xcolor}
\usepackage{appendix}
\usepackage{amsmath}
\usepackage{tasks}
\def\mathunderline#1#2{\color{#1}\underline{{\color{black}#2}}\color{black}}

\newenvironment{subproof}[1][\proofname]{%
  \begin{proof}[#1]%
}{%
  \end{proof}%
}

\newcommand\myeq{\mathrel{\overset{\makebox[0pt]{\mbox{\normalfont\tiny\sffamily def}}}{=}}}
\makeatletter
\usepackage{dcolumn}
\usepackage{soul}
\usepackage{bm}
\usepackage{amsthm}
\usepackage{enumitem}
\usepackage{amsmath}
\newtheorem{theorem}{Theorem}[section]
\newtheorem{lemma}{Lemma}

\usepackage{xr-hyper}
\usepackage{hyperref}
\usepackage{soul}
\usepackage[caption=false]{subfig}
\hypersetup{colorlinks,
	citecolor=blue
}
\usepackage{cleveref}
\usepackage{apptools}
\AtAppendix{\counterwithin{lemma}{section}}

\makeatletter
\newcommand*{\addFileDependency}[1]{
  \typeout{(#1)}
  \@addtofilelist{#1}
  \IfFileExists{#1}{}{\typeout{No file #1.}}
}
\makeatother

\begin{document}

\preprint{APS/123-QED}

\title{Polynomially efficient quantum enabled variational Monte Carlo for training neural-network quantum states for physico-chemical applications}

\author{Manas Sajjan$^\dagger$}
\affiliation{Department of Chemistry, Purdue University, West Lafayette, IN 47907}

\affiliation{Department of Electrical and Computer Engineering, North Carolina State University, Raleigh, NC 27606}

\author{Vinit Singh$^\dagger$}
\affiliation{Department of Chemistry, Purdue University, West Lafayette, IN 47907}

\affiliation{Department of Electrical and Computer Engineering, North Carolina State University, Raleigh, NC 27606}

\author{Sabre Kais}
\email{skais@ncsu.edu}

\affiliation{Department of Electrical and Computer Engineering, North Carolina State University, Raleigh, NC 27606}

\begin{abstract}

{\color{black} With several diverse architectures and astounding expressibility, neural-network quantum states (NQS) have emerged as a compelling alternative to traditional variational ansätze for simulating physical systems across diverse applications. These models, especially energy-based frameworks like Hopfield networks and Restricted Boltzmann Machines, draw from principles of statistical physics to map accessible quantum states onto an energy landscape, effectively serving as associative memory descriptors. In this work, we demonstrate that such energy-based models can be trained efficiently using Monte Carlo techniques enhanced by quantum devices. Our proposed algorithm scales linearly with circuit width and circuit-depth, constant in measurement count, devoid of any mid-circuit measurements, and even polynomial in storage, ensuring optimal computational efficiency. Our method also holistically works for both phase and amplitude fields of quantum states, which enhances the scope of the trial space drastically compared to previously designed techniques. Unlike classical approaches, sampling via quantum devices significantly reduces the convergence time of the underlying Markov Chain and produces more faithful sample estimates, highlighting the advantage of quantum-assisted training. We showcase the applicability of this method by accurately learning the ground states of both local spin models in quantum magnetism and non-local electronic structure Hamiltonians, even in distorted molecular geometries wherein strong multi-reference correlations dominate. Benchmarking these results against traditional methods reveals a high degree of agreement in all cases, underscoring the robustness of our approach. This work accentuates the acute potential for symbiotic collaborations between powerful machine learning protocols and near-term quantum devices to tackle diverse challenges in quantum state learning, making them highly relevant for applications in theoretical chemistry and condensed matter physics.}
\end{abstract}

\maketitle

\def\thefootnote{$\dagger$}\footnotetext{\text{These authors contributed equally to this work}}\def\thefootnote{\arabic{footnote}}


\maketitle

\section{Introduction}

{\color{black}Efficiently simulating stationary quantum states of matter has been one of the principal tasks of computational physics and chemistry. Beyond understanding the energy eigenspace, such states also form a convenient basis for treating any subsequent dynamical evolution. The primary challenge associated with the task is the well-known curse of dimensionality, which at its core is an unavoidable but undesirable inheritance from the geometric structure of the Kronecker product space of many-body quantum mechanics.
The latter renders numerical recipes reliant on exact solvability impracticable beyond a certain number of interacting degrees of freedom. Thus aligned with the objective of efficiently creating approximate target quantum states that inhabit an otherwise exponentially large Hilbert space, efforts have been made over the past few years to develop expressible representations of such states using neural networks \cite{carleo2017solving,medvidovic2024neural} which are provably known to be excellent function approximators \cite{cybenko1989approximation}. Commonly grouped under the ever-expanding umbrella of neural quantum states (NQS), a diverse variety of network architectures like Restricted Boltzmann Machines (RBM)\cite{torlai2018neural, carleo2017solving, deng2017quantum, nagy2019variational}, Auto-regressive neural networks (ARN) \cite{barrett2022autoregressive}, Deep-Boltzmann machines (DBMs) \cite{nomura2021purifying}, Recurrent neural networks (RNN) \cite{hibat2020recurrent,wu2023tensor}, simple feed-forward neural networks (DNN) \cite{jia2019quantum}, Fermionic neural networks \cite{robledo2022fermionic,liu2024unifying}, and even Transformers \cite{zhang2023transformer} have now been reported in literature to functionally represent a quantum state in tasks not only in learning energy eigenpairs \cite{melko2019restricted,medvidovic2024neural} but also in dynamical evolution of open and closed systems \cite{nagy2019variational, donatella2023dynamics,luo2022autoregressive}. The usual workflow in such architectures is to represent the target state on the basis of bit strings/spin configurations. For each such spin configuration fed as input, the network is trained with tunable parameters to produce the corresponding phase and amplitude of the target quantum state as output. It also has been proven analytically that any quantum state represented by an efficiently contractible tensor network (TN), such as matrix product states (MPS) and projected entangled pair states (PEPS), can also be represented by an NQS efficiently with polynomial resources, but the converse is strictly not true \cite{sharir2022neural} thus making commonly used TNs a subset of NQS in terms of expressive power. Emboldened by such assertions, recently deep neural network representations have been used to treat even large 2D models like the $10\times10$ non-frustrated Heisenberg model and even frustrated $J_1-J_2$ model of size $20 \times 20$ on square lattice and of size $18 \times 18$ on triangular lattices \cite{chen2024empowering} efficiently with remarkable accuracy. Even shallow NQS like RBM, which has been analytically proven to be a universal approximator for arbitrary probability distributions \cite{le2008representational,rrapaj2021exact}, have seen tremendous success in modeling volume-law entangled quantum states \cite{deng2017quantum,sajjan2023imaginary}, systems with higher-dimensional local Hilbert spaces \cite{vieijra2020restricted,pei2021neural,lahiri2021repesentation}, studying phase transitions in strongly correlated topological phases \cite{lu2019efficient, doi:10.1063/5.0128283}, fermionic systems like 2D materials and molecules \cite{sajjan2021quantum,sajjan2022quantum,sureshbabu2021implementation, choo2020fermionic}, bosonic systems \cite{mcbrian2019ground,pei2024specialising}, lattice gauge models \cite{luo2021gauge} and even steady states of dynamical evolutions generated by a Liouvillian \cite{nagy2019variational,cao2024neural,hryniuk2024tensor}.

Classically, expectation values of observables in these neural network-based representations are usually estimated by sampling important configurations / bit strings using Markov Chain Monte Carlo (MCMC) techniques \cite{dellaportas2003introduction,liu2001monte}. This operates by drawing samples from a prior/proposal distribution to construct a Markov chain that has the target distribution of the NQS as its steady state. Without domain knowledge about the target distribution, priors commonly used include local update wherein mutation of the currently chosen sample on a randomly chosen single site is performed, or global updates like choosing a configuration uniformly randomly \cite{medvidovic2024neural,vivas2022neural,hermann2023ab} each time. However, priors based on local updates are  susceptible to sample wastage and slow convergence,and uniform sampling is too random to be efficient, as it misses information about the geography of the landscape entirely. Cluster update schemes wherein a subset of spins are mutated also exists but are usually problem specific \cite{wolff1989collective}.
One may ask if a tailored quantum-enabled prior distribution can ameliorate some of these issues and make training advantageous. This is akin to a pertinent question that has been asked for several years now whether machine learning tasks on classical data can be accelerated using quantum hardware\cite{cerezo2022challenges,sajjan2022quantum}. The vision for designing a quantum-enabled training of NQS can be viewed as the diametrically opposite end of the spectrum where an algorithm rooted in machine learning is used to learn features of quantum data augmented by the power of quantum devices. Despite the aforementioned unprecedented representational capacity of NQS, this possibility is criminally understudied.

Some previous efforts of the authors attempted to develop a parameterized quantum circuit to sample the desired probability distribution for training RBM, a specific NQS \cite{Xia_2018,sajjan2021quantum}. The workflow was directly motivated by the suite of quantum-assisted algorithms developed over the past decade that have leveraged the variational principle to work with a low-depth parameterized circuit(ansatz), which can scout a smaller and more manageable subspace of the full Hilbert space\cite{cerezo2021variational}. However, the algorithm so designed, apart from being resource intensive in terms of gate depth and qubit requirements, was also limited to only using quantum-assisted training to model the amplitude field of the target state\cite{xia2018quantum}. The accompanying phase information was obtained entirely by classical means\cite{sajjan2021quantum}. The specific design of the protocol also required storing the full distribution, which leads to an implicit exponential run-time and storage cost. A part of the problem was that the circuit so constructed was for a completely different workflow compared to the usual variational quantum algorithms. For the specific NQS used, the purpose of the circuit was to facilitate the preparation of a targeted distribution that required performing non-unitary operations executed through mid-circuit measurements and ancillary reuse \cite{xia2018quantum,decross2023qubit}. This greatly extended the execution time of the circuit due to the wastage of shots, which, in turn, degraded precision due to measurement-induced errors. However, in this manuscript, we shall report the construction of a more generalizable protocol that can work for many different NQS and, as we shall show, cuts down on the resource requirements drastically, obviates the need for mid-circuit measurements, works holistically for both amplitude and phase information of the target wavefunction, is strictly polynomial in both runtime and storage, and also performs better than usual classical sampling strategies.

The major contributions in this manuscript are the following:
i) We introduce the concept of a surrogate neural network for training an NQS. Although the specific NQS used to demonstrate our protocol for learning an $n$-qubit state is RBM, the probability distribution on which the concomitant surrogate network is trained is analytically proven to approximate any arbitrary probability density defined over a finite support, thus imparting generalizability. We also design a polynomially efficient protocol to define this distribution through data-driven fitting and characterize any error due to truncation. ii) Furthermore, we analytically prove that under specific circumstances, training the target NQS using the distribution of the surrogate can be advantageous in terms of reducing the variance of the sample estimate of the observable. iii) We also show how a quantum circuit based on Hamiltonian simulation can be used to design a distribution sampling protocol specifically tailored to the surrogate and then numerically illustrate that samples from the circuit show faster convergence and less initial burn-in and auto-correlation, yielding a better quality converged steady distribution compared to well-known classical priors and even Haar random states. This subsequently enhances the accuracy of any sample estimates. One must note that even though the training of the NQS is variational, it is being trained classically with inputs generated from the quantum circuit. This sets our workflow different from other variational quantum algorithms as it is the parameterized quantum circuit in those cases which serves as a quantum-neural network \cite{PRXQuantum.2.040337} directly. In our case, we use a neural network hosted classically with an existence independent from that of the quantum circuit. We query the quantum circuit only as a sampler to extract $n$-bits of information from it each time in conformity with Holevo's bound \cite{holevo1973bounds, holevo1998quantum} and use that to train the classical network. In fact, a better grouping of our algorithm would be with the set of protocols designed for variational Monte Carlo simulation enabled by quantum circuits, which is an extremely underdeveloped territory and is only beginning to receive some attention\cite{xu2023quantum, kanno2024quantum,zhang2022quantum,mazzola2024quantum,huggins2022unbiasing}. Our protocol can also be envisioned as a step in the direction of recent analogous post-variational algorithms like subspace quantum diagonalization \cite{robledo2024chemistry,barison2024quantum}, or machine learning algorithms trained with classical shadows \cite{jerbi2024shadows,heidari2023quantum,li2021vsql} where samples from the quantum circuit are used and hence is a specific instance of the recently described CSIM-QE workflow \cite{cerezo2023does}. (iv) The circuit so prepared is rooted in Trotterization \cite{childs2021theory} as the specific choice of Hamiltonian simulation and hence can be systematically improved depending on preset user threshold for error. Also, it requires $O(n)$ qubits and a gate depth/layers of $O(\tau n)$ where $\tau$ is the time step interval for the Trotterization scheme used, which is usually independent of system size. Besides, there is no requirement for any mid-circuit measurement, and number of queries made to the circuit/number of measurements is $O(N_s)$, which is not only amenable to parallel execution through GPU accelerated learning protocols but also independent of system size and is preset by the user depending on precision of sample estimates. Also, the algorithm holistically accounts for both the phase and amplitude field of the target quantum state, obviating the need for any classical inputs. This vastly enhances the scope of systems that can be efficiently treated. The storage of the algorithm is also $\text{poly}(n)$ without the requirement of ever retrieving the full distribution, which has been formally proven to be always inefficient unless polynomial hierarchy collapses\cite{long2010restricted}. 
All of these features have to be contrasted with the previous algorithm for quantum training of an NQS (See Section IV for details). To the best of our knowledge, this is the best-known quantum-enabled algorithm for NQS training in terms of resources required. (v) The algorithm is exemplified using systems ranging from not only spin models but also electronic states of molecules under both equilibrium and distorted geometries. Despite its low implementation cost, in all cases, it has been found to produce highly accurate results below the chemical accuracy threshold, especially when combined with zero-variance extrapolation, which is inexpensive in our method and doesn't incur any additional overhead.

The remainder of the manuscript delves into the specific details of these improvements. The following section (Section II) provides a detailed discussion on the architecture of the NQS used and the surrogate. This is followed in Section III by the entire algorithmic workflow along with the quantum-enabled MCMC sampling. Several performance indicators/metrics are also analyzed here to compare them with other sampling strategies. The resource requirements of the algorithm are explained in Section IV and compared with previous approaches. Several applications of quantum state learning that are used for benchmarking the algorithm are shown in Section V, followed by concluding remarks and future directions in Section VI. }

\section{Architecture of the parent and surrogate network}

\begin{figure}[h!]
    \centering
    \includegraphics[width=1.02\linewidth]{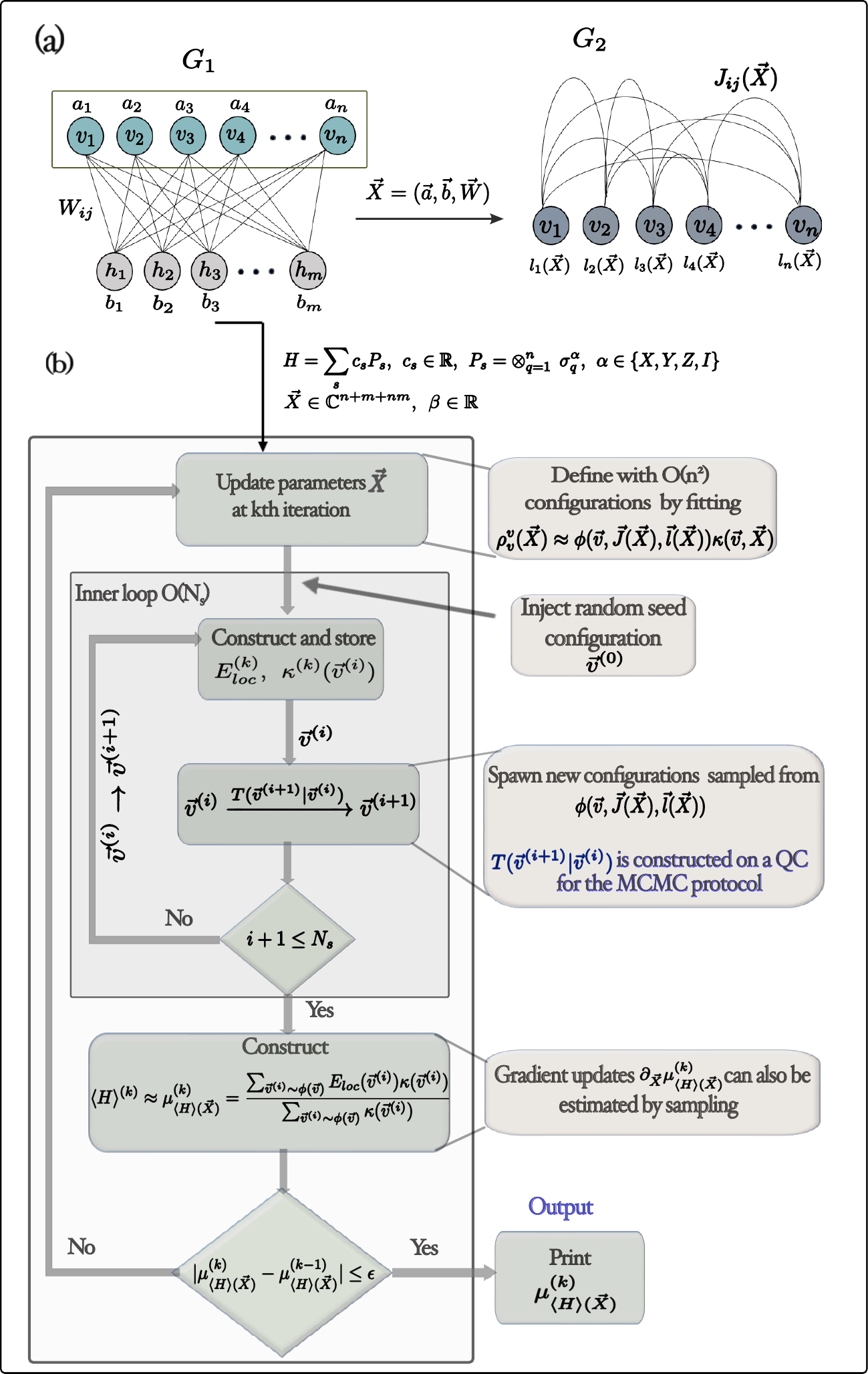}
    \caption{(a) The initial neural-network quantum state, shown on the left, is a Restricted Boltzmann Machine (RBM), a bipartite graph $G_1$ with visible neurons $v_i$ and hidden neurons $h_j$. Bias vectors $\vec{a}$ and $\vec{b}$ are associated with each layer, and their connections are parameterized by $\vec{W}$. Together, the parameter vector $X=(\vec{a},\vec{b},\vec{W}) \in \mathbb{C}^{m+n+mn}$ defines $G_1$ during training. On the right is the surrogate network $G_2$, constructed by our algorithm using $X$. $G_2$ consists of a single layer of visible neurons $v_i$ with all-to-all connections parameterized by $\vec{J}(\vec{X}) \in \mathbb{R}^{n^2}$ and biases $\vec{l}(\vec{X}) \in \mathbb{R}^n$. Our algorithm explicitly computes $(\vec{l}(\vec{X}), \vec{J}(\vec{X}))$ from $\vec{X}$.(b) The flowchart of the algorithm is shown for the hamiltonian $H$ of the driver (see text). We shall input the driver Hamiltonian as $H=\sum_s c_s P_s$, with $P_s = \otimes_{q=1}^n \sigma_\alpha^{(q)}, : \alpha \in \{x,y,z,0\}$, an initial parameter set $X$, and $\beta \in \mathbb{R}$, the algorithm constructs $G_2$ and define $E_{loc}(\vec{X})$, $\kappa(\vec{v})$ (See Eq.\ref{eq:loc_energy} and Eq.\ref{eq:surrogate_state}), draws samples using quantum-assisted methods, and outputs a self-converged sample estimate $\mu_{\langle H \rangle(\vec{X})} \approx \langle H \rangle$ (see Eq.\ref{eq:sample_mean_energy}). The same procedure provides gradient estimates (see Eq.\ref{eq:loc_grad}), which can be used to train $G_1$.}
    \label{fig:Scheme-flowchart}
\end{figure}

In Fig \ref{fig:Scheme-flowchart}(a) we define the two NQS of interest. On the left of the said figure is a bipartite graph $G_1=(V_1, E_1)$ consisting of two inter-connected spin registers. The vertex set $V_1=\{v\}_{i=1}^{i=n} \bigcup  \{h\}_{j=1}^{i=m}\,\,\rm{with}\:\: $(n,m)$ \:\in \:\mathbb{Z}_{+}$. The subset $\{v\}_{i=1}^{i=n}$ are called visible neurons with $v_i \in \{1,-1\} \:\:\forall \:\: i$ analogous to classical spins. A similar description holds for the subset $\{h\}_{j=1}^{j=p}$, which are collectively called hidden neurons. Each such neuron is endowed with a bias term (akin to an on-site magnetic field), which is denoted as $a_i$ for the set of visible neurons and $b_j$ for the set of hidden neurons. The set of edges $|E|=n*m$ defining connections between two subsets of neurons is parameterized by $W_{ij}$. Collectively, $X=(\vec{a}, \vec{b}, \vec{W}) \in \mathbb{F}^{n+m+nm}$ define the set of trainable parameters of the network where $\mathbb{F}$ is an arbitrary field of scalar variables. The Hamiltonian of the graph $G_1$ is that of a classical Ising model defined as 
\begin{align}
    \mathcal{H}(\vec{X}, \vec{v}, \vec{h}) &= \sum_{i=1}^{n} a_i v_i + \hspace{-0.07in}\sum_{j=1}^{m} b_j h_j + \hspace{-0.13in}\sum_{i=1,j=1}^{n,m} W_{ij} v_i h_j
    \label{eq: Ising_energy}
\end{align} 

When $\mathbb{F}= \mathbb{R}$, then one can define a classically correlated thermal state for the network \cite{smolensky1986information}, which has been used as a generative model in the literature under the name of Restricted Boltzmann Machine \cite{ackley1985learning, hinton2006reducing,larochelle2008classification, salakhutdinov2007restricted}. It is capable of learning an arbitrary discrete probability distribution\cite{le2008representational}. Similarly when $\mathbb{F}= \mathbb{C}$, the above network is capable of representing a quantum state \cite{carleo2017solving,torlai2020machine,melko2019restricted,PhysRevLett.122.250501,borin2020approximating,cao2024neuralnetworkapproachnonmarkovian,yoshioka2019constructing,hartmann2019neural} by defining a reduced density matrix over visible neurons (highlighted within a black outlined box in Fig.\ref{eq: Ising_energy}(a)) as the follows 
\begin{align}
\rho^v_{v^\prime}(\vec{X}) \propto e^{-\beta \sum_i a_i v_i + \sum_i a_i^* v_i^\prime} \:\prod_{j=1}^m \Gamma_j (\vec{v},\vec{b}, \vec{W}) \:\Gamma_j (\vec{v}^\prime,\vec{b^*}, \vec{W^*}) \label{eq:Ising_state}
\end{align} 

where $\Gamma_j (\vec{v},\vec{b}, \vec{W})=
\rm{Cosh}(\beta b_j + \beta \sum_{i=1}^n W_{ij} v_i)$.
One can use Eq.\ref{eq:Ising_state} as a proxy state/ansatz to learn a target quantum state given a driver Hamiltonian $H \in \mathcal{L}(\mathbb{C}^{2^n}$) (any system of chemical and physical importance) with $\vec{X}$ being variationally trained in the process \cite{PhysRevX.7.021021,yang2020artificial,hagai2023artificial, pan2024efficiency,Xia_2018,PhysRevB.101.195141,sajjan2021quantum,doi:10.1063/5.0128283}. As already mentioned, previous work by the authors\cite{xia2018quantum, sajjan2021quantum, sureshbabu2021implementation} have attempted to design a quantum circuit for training Eq.\ref{eq:Ising_state}. However, it was restricted to $\mathbb{F}= \mathbb{R}$ and hence could only model amplitude field with the relative phases of the configurational basis states being handled classically. The design of the circuit required a heavy dependence on number of hidden neurons $m$ with mid-circuit measurements powered on repeat-until-success modus operandi. This naturally led to shot wastages which enhances QPU time and slows execution. Besides, storage was also exponential in the scheme. Later efforts have tried to solve some of the problems by unitarizing graph $G_1$\cite{hsieh2021unitary} which can hamper the expressive qualities of the representation. In addition, the circuit so designed is specific to NQS of the form of $G_1$. Herein we design a new algorithm that retains the advantages of quantum-enabled training, keeping the run-time and storage strictly polynomial (circuit requirements are linear i.e. $O(n)$) end-to-end and also using $\mathbb{F}= \mathbb{C}$ and even readily extendible to other NQS designs. To motivate such a construction, we shall first establish the following theorem

\begin{theorem}\label{thm:factor_theo_main}
Given $\mathbb{V}=\{\vec{v}| \vec{v} \in \{-1,1\}^n\}$ where $n\in \mathbb{Z}_{+}$ and an arbitrary discrete probability distribution $P:\mathbb{V}\to [0,1]$, then the following statement is true
\begin{enumerate}
\item It is always possible to find a representation of $P(\vec{v})$ as 
\begin{eqnarray}
    P(\vec{v}) &=& \mathcal{N}\kappa(\vec{v}) \frac{e^{\sum_{\vec{a} \in \{0, 1\}^n, H(\vec{a})\le k} C_{\vec{a}} F_{\vec{a}}(\vec{v})}}{\sum_{\vec{v}}e^{\sum_{\vec{a} \in \{0, 1\}^n, H(\vec{a})\le k} C_{\vec{a}} F_{\vec{a}}(\vec{v})}} \nonumber \\
    &=& \mathcal{N}\kappa(\vec{v})\phi(\vec{v})
\end{eqnarray}
where $F_{\vec{a}}(\vec{v}) = \prod_i v_i^{a_i} \:\:\:, \vec{a}=(a_1, a_2....a_n) \in \{0,1\}^n ,\:\: \:\: \vec{v}=(v_1,v_2,...v_n)\:\:\in \mathbb{V}$, $\:\:\:C_{\vec{a}} \in \mathbb{R}$, and $H(\vec{a})=\sum_i^n a_i$ is the Hamming weight of $\vec{a}$. We have defined $\phi(\vec{v})=\frac{e^{\sum_{\vec{a} \in \{0, 1\}^n, H(\vec{a})\le k} C_{\vec{a}} F_{\vec{a}}(\vec{v})}}{\sum_{\vec{v}}e^{\sum_{\vec{a} \in \{0, 1\}^n, H(\vec{a})\le k} C_{\vec{a}} F_{\vec{a}}(\vec{v})}}$ as another discrete distribution $\phi:\mathbb{V}\to [0,1]$ which will be the backbone of a secondary surrogate network.
Also $\mathcal{N}\ne f(\vec{v}) \ge 0$ and $k\in \mathbb{Z}_{+}$, is a user-defined preset non-negative integer. $k$ controls the degree of the polynomial chosen for expressing $\phi(\vec{v})$.

\item $\kappa(\vec{v})$ is a configuration dependant non-negative prefactor. In the most generic case, $\exists$ $M \in \mathbb{R}\ge 0$ such that $|\log(\kappa(\vec{v}))|$ is upper bounded as  
\begin{eqnarray}
    |\log(\kappa(\vec{v}))| \le M\left(2^{n} -\sum_{j=0}^{k} \ ^nC_j \right)
\end{eqnarray}
\end{enumerate}
\end{theorem}

\begin{proof}
See Section B in Appendix
\end{proof}

It must be emphasized that the above assertion is general enough to model any discrete distribution $P(\vec{v})$ which may not even be of the form generated from graph $G_1$. The distribution $\phi(\vec{v})$ defined above can be envisioned to be associated with a surrogate network defined with $n$ spins (similar to the visible layer of graph $G_1$) which we introduce in this manuscript. Apart from the proof of existence above, we also provide a constructive data-driven protocol to find the parameters $\{C_{\vec{a}}\}_{\vec{a}\in{0,1}, H(\vec{a})\le k}$ defining the distribution $\phi(\vec{v})$. This allows the subsequent design of our algorithm to be conveniently extended to any NQS and construct its surrogate. The purpose of choosing this surrogate is that no matter how so ever complicated $P(\vec{v})$ is, the functional form of which may even not be precisely known, the functional form of $\phi(\vec{v})$ is always simple, precisely determinable (upto normalization) and hence amenable to being efficiently sampled. The prefactor $\kappa(\vec{v})$ provides additional expressibility/flexibility in modeling the target distribution $P(\vec{v})$ given the surrogate distribution $\phi(\vec{v})$ is represented with a finite order $k$. Exact computation of $\kappa(\vec{v})$ would require us to know the functional form of $P(\vec{v})$ and/or solve for the coefficients of a $(n-k)$ degree polynomial (see Section B in Appendix) which can be difficult for high $n$ and small $k$. So, two approaches can be used for its construction. If the functional form of $P(\vec{v})$ is known upto a normalization (the case here), one can try to make $\phi(\vec{v})$ mimic $P(\vec{v})$ closely and incorporate all deviations within $\kappa(\vec{v})$. Additionally, if that information is unknown (just samples from $P(\vec{v})$ is only accessible) one can use domain knowledge and parameterize $\kappa(\vec{v})$ in different functional forms which may even reap unforeseen advantages in constructing the estimated observable. One such parameterization specific to the context at hand is presented in Section C in Appendix and shows the advantage of reduction in variance of the sampled estimate in certain regimes. 

In this report we use the former of the two strategy for $\kappa(\vec{v})$ and choose $P(\vec{v})=\rho^v_v(\vec{X})$ (see Eq.\ref{eq:Ising_state}) as an illustration for manipulating a distribution generated from an NQS. For the surrogate network distribution $\phi(\vec{v})$ we truncate the polynomial in Theorem\ref{thm:factor_theo_main}(1) to second order i.e. use $k=2$. This can be visually represented as an additional Ising network $G_2=(V_2, E_2)$ (see top-right in Fig.\ref{fig:Scheme-flowchart}(a)). This new network is a fully-connected Ising graph consisting of $n$ neurons. The on-site bias terms are denoted as $l(\vec{X}) \in \mathbb{R}^n$ and the interconnections in edge set $E_2$ are denoted as $\vec{J}(\vec{X}) \in \mathbb{R}^{n^2}$. The state of this network is then defined by a distribution as $\phi_{\vec{v}}(\vec{l}(\vec{X}), \vec{J}(\vec{X})) \propto e^{-\beta \sum_i l_i(\vec{X}) v_i + \sum_{ij} J_{ij}(\vec{X}) v_i v_j}$. The parameters of this surrogate network $(\vec{l}(\vec{X}), \vec{J}(\vec{X}))$ are functionally related to the parent network $G_1$ as the distribution  $(\phi_{\vec{v}}, \rho^v_v)$ satisfies the following
    \begin{align}
     \rho^v_v(\vec{X}) \propto \kappa(\vec{v}, \vec{X})\phi_{\vec{v}}(\vec{l}(\vec{X}), \vec{J}(\vec{X})) \label{eq:surrogate_state}
    \end{align}
where $\kappa(\vec{v}, \vec{X})$ as before is a configuration dependent pre-factor chosen to establish the equivalence. This renders $\phi_{\vec{v}}(\vec{l}(\vec{X}), \vec{J}(\vec{X}))$ as a close approximant to $\rho^v_v(\vec{X})$.
    
\section{Algorithm}
Our protocol (see Fig.\ref{fig:Scheme-flowchart}(b)) is as follows
\begin{enumerate}
    \item In the first step, the parameters of the surrogate network $G_2$ which are $(\vec{l}(\vec{X}), \vec{J}(\vec{X}))$ are procured by fitting $(\phi_{\vec{v}}, \rho^v_v)$ pairs for just $O(n^2)$ configurations. This is completely a data-driven construction and doesnt require normalization for either $\rho^v_v(\vec{X})$ or 
    $\phi_{\vec{v}}(\vec{l}(\vec{X}), \vec{J}(\vec{X}))$. The algorithm works by choosing $q$ largest configurations of $\rho^v_v(\vec{X})$ in $O(qn)$ runtime which can be further optimized using specific data-structures. This step completely defines the parameter set $(\vec{l}(\vec{X}), \vec{J}(\vec{X}))$ for the distribution $\phi_{\vec{v}}$ in $G_2$ and makes it suitable to draw samples from. Further details about the fitting procedure can be found in Section D in Appendix.

    \item Given a Pauli-decomposed form of a driver hamiltonian $H \in \mathcal{L}(\mathbb{C}^{2^n}) = \sum_s c_s P_s$ where $P_s = \otimes_{q=1}^n \sigma_\alpha^{(q)}, \: \alpha \in \{x,y,z,0\}$, one can define local energy as follows
    \begin{align}
        \langle H \rangle (\vec{X}) = \frac{Tr(\rho(\vec{X}) H)}{Tr(\rho(\vec{X}))} = \sum_{\vec{v}} \frac{\rho^v_v(\vec{X}) E_{loc}(\vec{v}, \vec{X})}{\sum_{\vec{v}} \rho^v_v(\vec{X})}
        \label{eq:loc_energy}
    \end{align}
    where $E_{loc}(\vec{v}, \vec{X}) = \frac{\sum_{\vec{v}^\prime} H^v_{v^\prime} \rho^v_{v^\prime}(\vec{X})}{\rho^v_v(\vec{X})}$ is the local energy.
     Using Eq.\ref{eq:surrogate_state}, one can replace $\rho^v_v$ in Eq.\ref{eq:loc_energy} to 
     obtain 
     \begin{align}
     \langle H \rangle (\vec{X}) \approx \frac{\sum_{\vec{v}} \kappa(\vec{v}, \vec{X}) \phi_{\vec{v}}(\vec{l}(\vec{X}), \vec{J}(\vec{X})) E_{loc}(\vec{v},\vec{X})}{\sum_{\vec{v}} \kappa(\vec{v}, \vec{X}) \phi_{\vec{v}}(\vec{l}(\vec{X}), \vec{J}(\vec{X}))}  \label{eq:avg_energy_surrogate}
     \end{align}
It is to be noted that Eq.\ref{eq:avg_energy_surrogate} is averaged over the distribution of the surrogate network with $\kappa(\vec{X}, \vec{v})$ being the kernel. Thus, the task would be to procure an unbiased sample estimate to Eq.\ref{eq:avg_energy_surrogate} hereafter denoted as 
\begin{align}
\mu_{\langle H \rangle(\vec{X})} =  \frac{\sum_{\vec{v} \sim \phi} \kappa(\vec{v},\vec{X})E_{loc}(\vec{v},\vec{X})}{\sum_{\vec{v} \sim \phi} \kappa(\vec{v}, \vec{X})}
\label{eq:sample_mean_energy}
\end{align}
The number of samples required to create this faithful estimate to a preset error threshold $\epsilon$ is $N_s = O\Big(\frac{\rm{Var}(\mu_{\langle H \rangle(\vec{X})})}{\epsilon^2}\Big)$ \cite{saw1984chebyshev,Stellato_2017} where $\rm{Var}(\mu_{\langle H \rangle(\vec{X})})$ is the population variance of the estimated expectation. 


\item As noted in the previous point to estimate $\mu_{\langle H \rangle(\vec{X})}$ in Eq.\ref{eq:sample_mean_energy}, one would need samples from $\phi_{\vec{v}}(\vec{l}(\vec{X}), \vec{J}(\vec{X}))$ which can be obtained using Metropolis-Hastings algorithm\cite{chib1995understanding} employing a proposal distribution by the user. Let the proposal distribution chosen by the user to sample $\vec{v}^{(i+1)}$ given $\vec{v}^{(i)}$ be denoted as $P_{prop}(\vec{v}^{(i+1)} | \vec{v}^{(i)})$. To ensure detailed balance, we define a transition probability matrix $T(\vec{v}^{(i+1)}|\vec{v}^{(i)})$ as the following:
\begin{align}
&T(\vec{v}^{(i+1)} | \vec{v}^{(i)})\nonumber \\ &= min\Bigg(1, \frac{\phi(\vec{v}^{(i+1)}) P_{prop}(\vec{v}^{(i+1)} | \vec{v}^{(i)})}{\phi(\vec{v}^{(i)}) P_{prop}(\vec{v}^{(i)}| \vec{v}^{(i+1)})}\Bigg) \nonumber\\ &\times P_{prop}(\vec{v}^{(i+1)}|\vec{v}^{(i)}) 
\label{eq:trans_matrix}
\end{align}

The protocol repeatedly samples from $T(\vec{v}^{(i+1)} | \vec{v}^{(i)})$ and draw a list of configurations $\vec{v}^{(i)} \rightarrow \vec{v}^{(i+1)} \rightarrow \vec{v}^{(i+2)}...$ using which $\mu_{\langle H \rangle(\vec{X})}$ in Eq.\ref{eq:sample_mean_energy} can be estimated.
The acceptance/rejection criteria built-in the definition of Eq.\ref{eq:trans_matrix} ensures that after an initial burn-in, configurations are drawn from 
the target $\phi_{\vec{v}}(\vec{l}(\vec{X}), \vec{J}(\vec{X}))$.
It must be emphasized that the dependence of $T(\vec{v}^{(i+1)} |\vec{v}^{(i)})$ on the ratio of $\frac{\phi(\vec{v}^{(i+1)})}{\phi(\vec{v}^{(i)})}$ obviates the need for the computation of the normalization constant which is resource-intensive. This is unlike direct sampling protocols (accept-reject sampling or inversion transform of the cumulative distribution \cite{derflinger2010random}) where normalization of $\phi_{\vec{v}}(\vec{l}(\vec{X}), \vec{J}(\vec{X}))$ would be deemed necessary thereby motivating our choice.
Several different proposal distributions $P_{prop}(\vec{v}^{(i+1)} |\vec{v}^{(i)})$ can be chosen by the user provided it is ergodic and leads to aperiodic and irreversible states. This guarantee a unique stationary distribution and convergence of the generated Markov chain to such a distribution \cite{liu2001monte, dellaportas2003introduction}. Each such proposal generates a new $T(\vec{v}^{(i+1)} | \vec{v}^{(i)})$ in accordance with Eq.\ref{eq:trans_matrix}. In this report, we draw comparisons between the performance of several such proposal distributions (See Fig.\ref{fig:2nd_panel}(a)), each of which is defined as follows:

\begin{figure*}
    \centering
    \includegraphics[width=1\linewidth]{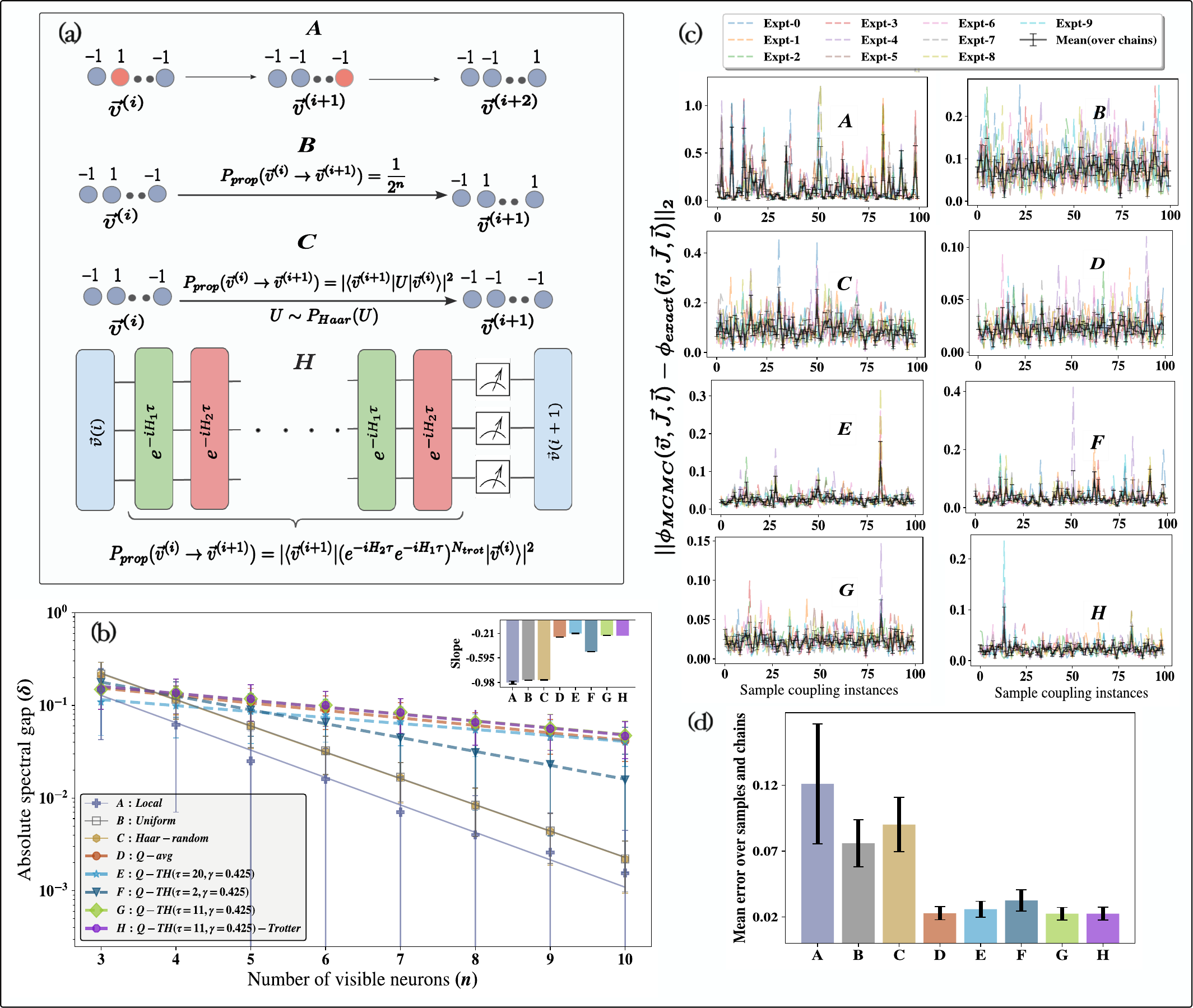}
    \caption{(a) Various proposal matrices for an $n$-qubit state are displayed. The local proposal (A) flips a single spin at a given site with $P_{prop}(\vec{v}^{(i+1)} |\vec{v}^{(i)})=\langle \vec{v}^{(i+1)}| \frac{\sum_k X^{(k)}}{n}|\vec{v}^{(i)}\rangle$. The uniform proposal (B) samples configurations equally with $P_{prop}(\vec{v}^{(i+1)} |\vec{v}^{(i)})=\frac{1}{2^n}$. The Haar random proposal (C) uses $P_{prop}(\vec{v}^{(i+1)} |\vec{v}^{(i)}) = |\langle \vec{v}^{(i+1)}| U|\vec{v}^{(i)}\rangle|^2$, where $U$ is a Haar random unitary. The quantum-enhanced proposal (H) samples via $P_{prop}(\vec{v}^{(i+1)} | \vec{v}^{(i)}) = |\langle \vec{v}^{(i+1)}| U(\tau, \gamma)|\vec{v}^{(i)}\rangle|^2$ with $U(\tau, \gamma) = e^{(-i \gamma h_1 + (1-\gamma)h_2)\tau}$, $h_1 = \sum_i l_i(\vec{X})\sigma_z^{(i)} + \sum_{i,j} J_{ij}(\vec{X})\sigma_z^{(i)}\sigma_z^{(j)}$, and $h_2 = \sum_i \sigma_x^{(i)}$ and $(\gamma,\tau) \in \mathbb{R}$. Variants $\mathrm{E-G}$ modify $(\gamma, \tau)$ to produce time-homogeneous proposals with single instantaneous realizations of $(\gamma, \tau)$ instead of an average as in $\mathrm{D}$. The output $\vec{v}^{(i)}$, a bit-string of length $O(n)$, is used for the next step. (b) The spectral gap $\delta = \lambda_0 - \lambda_1$ of $T(\vec{v}^{(i+1)} |\vec{v}^{(i)})$ is analyzed for $\mathrm{A-H}$. Quantum proposals $\mathrm{D-H}$ show slower gap decay than classical $\mathrm{A-C}$, improving mixing time which results in faster convergence to the steady distribution.(c) For $n=8$, samples are drawn for each $\mathrm{A-H}$, and deviations from the exact distribution are measured as $||\phi_{MCMC}^{(\vec{l}, \vec{J})}(\vec{v})-\phi_{exact}^{(\vec{l}, \vec{J})}(\vec{v})||_2$ over 100 random instances of $(\vec{l}, \vec{J})$. From each such instance, 10,000 samples are used. Results from 10 Markov chains are shown as colored traces, with mean and standard deviation as black traces. (d) We take the mean of the black trace in (c) above over all coupling instances for each proposal and display the resuls here with the errorbars denoting the respective standard deviations. It shows quantum proposals $\mathrm{D-H}$ reduce average error by a factor of 5 compared to $\mathrm{A-C}$, with lower variaance, ensuring robustness across many instances of $(\vec{l}, \vec{J})$.}

    \label{fig:2nd_panel}
\end{figure*}

\begin{itemize}
    \item \textbf{A = Local Proposal} $\implies$ $P_{prop}(\vec{v}^{(i+1)} | \vec{v}^{(i)}) = \langle \vec{v}^{(i+1)}| \frac{\sum_k X^{(k)}}{n}|\vec{v}^{(i)}\rangle$.
    This corresponds to flipping/mutating a single classical spin within $\vec{v}^{(i)}$ to spawn a new configuration $\vec{v}^{(i+1)}$ and is quite popularly used in literature\cite{PhysRevB.107.205102, 10.21468/SciPostPhysCore.6.4.088, PhysRevB.107.195115}
    
    \item \textbf{B = Uniform Proposal} $\implies$ $P_{prop}(\vec{v}^{(i+1)} | \vec{v}^{(i)}) = \frac{1}{2^n}$.
    
    \item \textbf{C = Haar Random Proposal} $\implies$ $P_{prop}(\vec{v}^{(i+1)} | \vec{v}^{(i)}) = |\langle \vec{v}^{(i+1)}| U|\vec{v}^{(i)}\rangle|^2$ where $U \sim P_{Haar}(U)$ i.e sampling a random unitary from the Haar measure. \cite{zyczkowski2000truncations, alagic2020efficient}
    
    \item \textbf{D = Quantum average Proposal} $\implies$ $P_{prop}(\vec{v}^{(i+1)} |\vec{v}^{(i)}) = \frac{1}{(\tau_2-\tau_1)(\gamma_2-\gamma_1)}$ $\int_{\tau=\tau_1}^{\tau_2} \int_{\gamma=\gamma_1}^{\gamma_2} |\langle \vec{v}^{(i+1)}| U(\tau, \gamma)|\vec{v}^{(i)}\rangle|^2 \: d\tau d\gamma$ where $U(\tau, \gamma) = e^{(-i \gamma h_1 + (1-\gamma)h_2)\tau}$. $h_1$ is inspired directly from the energy of the thermal distribution $\phi_{\vec{v}}(\vec{l}(\vec{X}), \vec{J}(\vec{X}))$ as $h_1=\sum_i l_i(\vec{X})\sigma_z^{(i)} + \sum_{i,j} J_{ij}(\vec{X})\sigma_z^{(i)}\sigma_z^{(j)}$ and $h_2$ is a mixer defined as $h_2 =\sum_i \sigma_x^{(i)}$. Both $(\gamma, \tau)$ are randomly sampled to approximate the integral and we set $(\gamma_1, \gamma_2)=(0.25,0.6)$ and $(\tau_1, \tau_2)=(4,20)$ \cite{layden2023quantum}. The unitary $U(\tau, \gamma)$ is exactly computed.
    
    \item \textbf{E = Quantum time-homogeneous Proposal}\textbf{$\:\:(\tau=20, \gamma=0.425)$} 
    $\implies$ $P_{prop}(\vec{v}^{(i+1)} | \vec{v}^{(i)})$ = $|\langle \vec{v}^{(i+1)}| U(\tau=20, \gamma=0.425)|\vec{v}^{(i)}\rangle|^2$ where $U(\tau, \gamma)$ similar to in $D$. Here $\tau$ is chosen to be the highest endpoint of its interval, and $\gamma$ is at the mid-point of its range in $D$. The unitary $U(\tau, \gamma)$ is exactly computed.
    
    \item \textbf{F = Quantum time-homogeneous Proposal}\textbf{$\:\:(\tau=2, \gamma=0.425)$} 
    $\implies$ $P_{prop}(\vec{v}^{(i+1)} | \vec{v}^{(i)})$ = $|\langle \vec{v}^{(i+1)}| U(\tau=2, \gamma=0.425)|\vec{v}^{(i)}\rangle|^2$ where $U(\tau, \gamma)$ similar to in $D$. Here $\tau$ is chosen to be the lowest endpoint of its interval, and $\gamma$ is at the mid-point of its range in $D$. The unitary $U(\tau, \gamma)$ is exactly computed.
    
    \item \textbf{G = Quantum time-homogeneous Proposal}\textbf{$\:\:(\tau=11, \gamma=0.425)$} 
    $\implies$ $P_{prop}(\vec{v}^{(i+1)} |\vec{v}^{(i)})$ = $|\langle \vec{v}^{(i+1)}| U(\tau=2, \gamma=0.425)|\vec{v}^{(i)}\rangle|^2$ where $U(\tau, \gamma)$ similar to in $D$.  Here $\tau$ is chosen to be the midpoint of its interval, and $\gamma$ is at the mid-point of its range in $D$. The unitary $U(\tau, \gamma)$ is exactly computed.
    
    \item \textbf{H = Quantum time-homogeneous Trotterized Proposal}\textbf{$\:\:(\tau=11, \gamma=0.425)$} 
    $\implies$ $P_{prop}(\vec{v}^{(i+1)} |\vec{v}^{(i)})$ = $|\langle \vec{v}^{(i+1)}| U(\tau=2, \gamma=0.425)|\vec{v}^{(i)}\rangle|^2$ with similar settings as in $G$. However, the unitary $U$ is not exactly computed, unlike in $D-G$, but is Trotterized to facilitate implementation on a quantum circuit, i.e., $U(\tau, \gamma) \approx (e^{-i (1-\gamma) h_2\delta t} e^{-i \gamma h_1\delta t})^{N_{trot}}$ with $N_{trot} = \frac{\tau}{\delta t}$. In the circuit implementation, each of the $n$ neurons in $G_2$ is replaced by a qubit.  As is illustrated in Fig.\ref{fig:2nd_panel}(a), the circuit starts with the incumbent instance of configuration $\vec{v}^{(i)}$ encoded within a state $|\vec{v}^{(i)}\rangle$ for the qubits. Within the said configuration, every element $-1$ corresponds to $|0\rangle$ for the respective neuron/qubit, and the element $1$ corresponds to $|1\rangle$. This state preparation is quite inexpensive and can be easily done using a single layer of $\sigma_{x}$ gates at appropriate neurons. The Trotterized gates for each layer are thereafter applied as $e^{-i (1-\gamma) h_2\delta t} e^{-i \gamma h_1\delta t})$ (See Section F in Appendix for a full-circuit decomposition) and repeated for $N_{trot}$ layers. The action of the pulse $e^{-i \gamma h_1\delta t}$ is responsible for changing configurations and swapping the amplitudes of various configurations once a superposition state is created. The pulse $e^{-i (1-\gamma) h_2\delta t}$ changes the energy of each configuration, which amounts to altering its relative phase within a superposition state. Both pulses are necessary to achieve the required interference pattern. The final output state is then sampled \textit{just once} to facilitate collapse in $|\vec{v}^{(i+1)}\rangle$ and retrieve a new configuration $\vec{v}^{(i+1)}$ which is supplied as the proposed candidate for the next iteration.
\end{itemize}

In Fig.\ref{fig:2nd_panel}(b), we compare the absolute spectral gap $\delta$ \cite{atchade2021approximate, layden2023quantum} for the transition probability matrix (see Eq.\ref{eq:trans_matrix}) defined using proposal distributions $A-H$ as defined above. $\delta = \lambda_0-\lambda_1$ where $\lambda_i$ is the $i$th eigenvalue of $T(\vec{v}^{(i+1)}| \vec{v}^{(i)})$ in a sorted list with $\lambda_i > \lambda_{i-1}\:\: \forall i$. The largest eigenvalue for $T(\vec{v}^{(i+1)} | \vec{v}^{(i)})$ is 1 corresponding to the stationary distribution which in this case is guaranteed to be $\phi_{\vec{v}}(\vec{l}(\vec{X}), \vec{J}(\vec{X}))$ by detailed balance. So $\delta$ essentially quantifies the inter-separation between the eigenvalue corresponding to the stationary distribution and the next highest one in the transition matrix. The usefulness of the quantity $\delta$ is quantified in the concept of mixing time \cite{bhatnagar2011computational} $(t_{mix})$, which roughly estimates the minimum number of samples drawn before the Markov chain converges to the target distribution with an error threshold (say $e$) in the total variation distance. 
The following bounds \cite{hsu2015mixing,layden2023quantum} inter-relate $\delta$ and $t_{mix}$
\begin{eqnarray}
\left( \delta^{-1} - 1 \right) \ln \left( \frac{1}{2e} \right) \leq t_{mix} \leq \delta^{-1} \ln \left( \frac{1}{e \min_{\vec{v}}\phi(\vec{v})} \right)
\label{eq:mix_time_delta}
\end{eqnarray}
It is clear that low $\frac{1}{\delta}$ (or high $\delta$) promotes lower mixing time. From Fig.\ref{fig:2nd_panel}(b), we see that $\delta$ is not only high for the quantum-enabled class of proposal distributions $D-H$ but also decreases 
slower by a factor of three (see inset of Fig.\ref{fig:2nd_panel}(b) for slopes) compared to classical proposals $A-C$. All points correspond to mean values (error bars are respective variances) over $250$ random instances of parameters $(\vec{l}(
\vec{X}),\vec{J}(\vec{X}))$ which defines the target distribution ${\phi(\vec{v}, \vec{X})}$. This decisive cubic advantage is the primary motivation for using a quantum-enhanced proposal distribution. Finally, to compare the quality of convergence and see the actual error in estimating the target distribution $\phi_{\vec{v}}(\vec{l}(\vec{X}), \vec{J}(\vec{X}))$, in Fig.\ref{fig:2nd_panel}(c) we compute the $l_2$-norm difference between the exact distribution (denoted in Fig.\ref{fig:2nd_panel}(c) as $\phi^{exact}_{\vec{v}}(\vec{l}, \vec{J})$)
and the estimated distribution constructed from the histogram of samples for each proposal $A-H$ (denoted in Fig.\ref{fig:2nd_panel}(c) as $\phi^{MCMC}_{\vec{v}}(\vec{l}, \vec{J})$). This is done for system $n=8$ neurons/qubits in network $G_2$ (see Fig.\ref{fig:Scheme-flowchart}(a)) for $100$ random instances of parameter vector $(\vec{l}(\vec{X}), \vec{J}(\vec{X}))$. For each parameter instance, we construct 10 Markov chains (labelled as Expt-$i \:\:\forall \:i \:\:\in \{0,9\}$), and each such chain has 10,000 sample configurations drawn. For each proposal, we plot the error in the convergence of the distribution averaged over all 10 Markov chains with error bars representing the respective variances. It is clear that proposals like $A, B, C$ (especially $A$) have higher mean and variance compared to quantum-enabled proposals $D-H$. This overall performance can be best understood from Fig.\ref{fig:2nd_panel}(d) where the bars show the same error in Fig.\ref{fig:2nd_panel}(c) averaged not only overall 10 chains but also overall $100$ parameter instances. We see that compared to classical proposals, all the quantum-enabled proposals $D-H$ not only have smaller errors but also have fewer fluctuations, indicating robustness to the claim.
It must be mentioned that all the proposals $A-H$ are symmetric i.e. $P_{prop}(\vec{v}^{(i+1)} |\vec{v}^{(i)}) = P_{prop}(\vec{v}^{(i)} |\vec{v}^{(i+1)})$ which precludes the need to include the factor 
$\frac{ P_{prop}(\vec{v}^{(i+1)} | \vec{v}^{(i)})}{P_{prop}(\vec{v}^{(i)}| \vec{v}^{(i+1)})})$ in Eq.\ref{eq:trans_matrix} explicitly. Among $D-H$, we see all the quantum-enabled proposals work similarly except $F$, highlighting decreasing the evolution time $\tau$ has a detrimental effect. However extreme enhancement in $\tau$ comes with the cost of deep quantum circuits as $N_{trot}=\frac{\tau}{\delta t}$
is enhanced, thereby making the execution error-prone for NISQ implementations. Besides, alteration in $\delta t$ to compensate for this increase would lead to higher Trotter errors. Keeping a balance between performance and cost for circuit implementation, we use the proposal $H$ (computed using $\delta t =0.2$) henceforth for all our computations. 


\item Once a configuration list $\{\vec{v}^{(i)}\}_{i=1}^{N_s}$ is available from the MCMC protocol in the above step, one can then compute a faithful estimate of $\mu_{\langle H \rangle(\vec{X})}$ for the incumbent instance of parameters $\vec{X}$ (and hence $(\vec{l}(\vec{X}), \vec{J}(\vec{X}))$). Since we obtain this estimate at the $k$-th iteration, let us denote this as 
$\mu_{\langle H \rangle(\vec{X})}^{(k)}$. This estimate is then compared with $\mu_{\langle H \rangle(\vec{X})}^{(k-1)}$, and if the difference is below a preset threshold, then the algorithm can return the current estimate, directly. If not, then the parameters of the network $G_1$ are upgraded for running the next iteration in the protocol. As used in this article, this upgrade is gradient-assisted. In fact gradients associated with each component $X_i$ denoted as $\partial_{X_i}\langle H \rangle (\vec{X})$ is analytical and can be expressed as 
\begin{align}
\partial_{X_i}\langle H \rangle (\vec{X}) = \langle(\mathcal{D}_{X_i} \odot H) \rangle -\langle \mathcal{D}_{X_i} \rangle \langle H \rangle \label{eq:loc_grad}
\end{align}
where all $\langle ....\rangle$ are computed over ${\phi(\vec{v}, \vec{X})}$ with a kernel $\kappa(\vec{v}, \vec{X})$. The $\odot$ in the first term represents element-wise Hadamard product, and $\langle H \rangle$ is the same as in Eq.\ref{eq:avg_energy_surrogate}. Each quantity can be replaced with their respective sample estimates computed using the configuration list $\{\vec{v}^{(i)}\}_{i=1}^{N_s}$. For further information and the definitions of the local operator $\mathcal{D}_{X_i}$, see Section E in Appendix . 
\end{enumerate} 


\section{Resource Requirements}

Given a construction of $G_1$ with $n$ neurons in the visible layer and $m$ in the hidden layer, the protocol starts with a random instance of the parameter vector $\vec{X}$. Since $\vec{X}\in \mathbb{C}^{n+m+nm}$, the number of independent variables to track is $O(2(nm +n+m))= O(nm)$ where the factor of 2 accounts for phase and amplitude of each component of $\vec{X}$. This instance of $\vec{X}$, defines elements of $\rho^v_v (\vec{X})$ in $G_1$ using which optimized parameter set $(\vec{l}(\vec{X}), \vec{J}(\vec{X}))$ of the surrogate $G_2$ is constructed. This step is achieved through fitting pairs $(\rho^v_v(\vec{X}), \phi_{\vec{v}}(\vec{l}(\vec{X}), \vec{J}(\vec{X})))$ using just $O(n^2)$ configurations and hence is polynomially efficient. The next step involves sampling configurations from $\phi_{\vec{v}}(\vec{l}(\vec{X}), \vec{J}(\vec{X}))$ using the quantum-enabled proposal $\rm{H}$. This step requires a quantum circuit of $n$-qubits, which is not only cheaper than the circuit width for the previous proposals \cite{sajjan2021quantum, Xia_2018}but also independent of $m$. The initial state preparation to encode a given seed configuration requires $O(n)$ single-qubit $\sigma_x$ gates on appropriate qubits, as said before, and all such gates can be executed in a parallel fashion within a single layer, thereby rendering the depth of this specific step to be $O(1)$. In the subsequent Trotterization scheme, one attempts to approximate the overall unitary $U(\tau, \gamma) = e^{-i (\gamma\sum_i l_i(\vec{X})\sigma_z^{(i)} + 
\gamma \sum_{i,j} J_{ij}(\vec{X})\sigma_z^{(i)}\sigma_z^{(j)} + 
(1-\gamma)\sum_i \sigma_x^{(i)})\tau}$ using $p$-th order product formulas with $N_{trot}$ steps. This incurs an error \cite{childs2021theory, yi2022spectral, PhysRevResearch.6.033309} of $O\Big(\frac{(\sum_{i=1}^{n^2+2n} ||\hat{\chi}_i||_1 \tau)^{p+1}}{N_{trot}^p}\Big)$ where $\hat{\chi}_i \in (\gamma l_i(\vec{X})\sigma_z^{(i)},\:\: \gamma J_{ij}(\vec{X})\sigma_z^{(i)}\sigma_z^{(j)},\:\: (1-\gamma)\sigma_x^{(i)})$ . To have this error bounded within a preset threshold of say $\epsilon_{trot}$ it suffices to have $N_{trot} \approx O\Big(\frac{(\sum_{i=1}^{n^2+2n} ||\hat{\chi}_i||_1 \tau)^{1+ \frac{1}{p+1}}}{\epsilon_{trot}^{\frac{1}{p}}}\Big)$. In our case, we use $p=1$, but the protocol is systematically improvable for higher product formulas at the cost of enhancing $N_{trot}$ as seen above \cite{childs2021theory}. For an arbitrary $p$, within each Trotter layer we have $O(n)$ single qubit gates and $O(n^2)$ two-qubit gates specifically of type $R_{zz}$ to simulate the effect of $e^{-i ( 
\gamma \sum_{i,j} J_{ij}(\vec{X})\sigma_z^{(i)}\sigma_z^{(j)})\tau}$.If CNOT is chosen as the basic entangling gate, then we show (see Section F in Appendix) that implementation of each such $R_{zz}$ requires two CNOTs. This makes the total number of CNOTs per Trotter layer $O(2n^2)$. Similar decomposition is also available in see Section F in Appendix if the native-two-qubit basis gate is an ECR. However, to estimate the runtime for parallel execution of quantum gates, one also has to estimate circuit depth and not just the total number of gates. Since in each Trotter layer, it is possible to implement all single-qubit gates in a parallel fashion, and they contribute a depth of $O(1)$ whereas out of the $O(n^2)$ two-qubit gates, one can implement $O(n/2)$ (if $n\%2=0$, otherwise $n/2 -1$) gates in a parallel fashion on non-contiguous qubits at a given time. This makes the circuit-depth to be linear, i.e., $O(n)$ in each Trotter layer. The circuit at the end is measured on a computational basis to retrieve a configuration that can be used to estimate local energy (see 
Eq.\ref{eq:loc_energy}) and hence $\mu_{\langle H \rangle (\vec{X})}$
(see Eq.\ref{eq:sample_mean_energy}). Apart from this, one will also estimate $O(2(nm+n+m))$ gradient terms corresponding to each parameter for subsequent updates. Thus, the total number of observables to measure is also $O(2(nm+n+m)+1)$. Each of these observable can be estimated with a precision of $\epsilon_{obs}$ by querying the above quantum circuit $O(N_s)$ times where $N_s \approx O(\frac{Var(\hat{O})}{\epsilon_{obs}^2})$  with $\hat{O}$ either being the local energy or the respective gradient observables. This query complexity of $O(N_s)$ is fixed and independent of system size $n$. This must be contrasted with the repeat-until success protocol and the wastage of shots due to mid-circuit conditional measurements in the previous scheme\cite{Xia_2018} where the successful collapse of the ancilla register in state $|1\rangle$ was essential to proceed. An attempt to avoid that and accelerate runtime necessitated extra qubits, making the circuit width $O(nm)$. In this scheme, all of these issues are completely alleviated. 


To compute storage complexity, one must emphasize that within each cycle of parameter update, the algorithm necessitates only storing $O(nm)$ parameters in $\vec{X}$ and $(\vec{l}(\vec{X}), \vec{J}(\vec{X}))$ combined, the corresponding gradient vector $\partial_{X_i} \langle H \rangle (\vec{X})$ having $O(nm)$ entries and a single estimate of local energy $\mu_{\langle H \rangle (\vec{X})}$ to assess self-convergence. Unlike in previous versions of quantum-enabled training of neural-network ansatz  \cite{Xia_2018, sajjan2021quantum}, the full quantum state is never retrieved from quantum memory (necessitating exponential storage) as phase information therein was injected classically. Herein, the parameters are intrinsically complex, making them inherently capable of treating the amplitude and phase field of the quantum state holistically in a unified fashion. One must also note 
even storing the proposal matrix $P_{prop}(\vec{v}^{i+1}|\vec{v}^{i})$ is not necessary due to the fact that is symmetric, thereby cancelling the factor $\frac{P_{prop}(\vec{v}^{i+1}|\vec{v}^{i})}{P_{prop}(\vec{v}^{i}|\vec{v}^{i+1})}$ in Eq.\ref{eq:trans_matrix}. Even for the computation of local energy, if the driver Hamiltonian $H=\sum_s^{\Gamma} c_s P_s$ 
where $P_s$ as defined before are Pauli words and $c_s \:\: \in \mathbb{R}$, then for a given input configuration $\vec{v^{(i)}}$, one would need at most $\Gamma$ different configurations (one for each $P_s$) to scan. It must be noted to compute such configurations; it is not necessary to store the driver $H \in \mathcal{L}(C^{2^n})$ in a dense format as its action on a given $\vec{v^{(i)}}$ can be adjudicated from the Pauli words $P_s$ analytically. Using this stored set of $O(\Gamma)$ configurations one then computes $\rho^v_{v^\prime}$ analytically (see Eq.\ref{eq:Ising_state}) and evaluate $E_{loc}(\vec{v}, \vec{X})$. Similar considerations apply for computing $\mathcal{D}_{X_i}$ in analytical gradient estimation. At the final step of the protocol, one needs to store $O(2nm+2n+2m+1)$ variables corresponding to the converged estimate of $\mu_{\langle H \rangle (\vec{X})}$ and converged parameter vector $\vec{X}$ for computing other properties of the target state in subsequent analysis.

\section{APPLICATIONS}
In this section, we shall study several applications wherein the ground state of a given driver Hamiltonian will be learned using the protocol delineated above. We shall train the parent and surrogate network by sampling from the Quantum time-homogeneous Trotterized proposal (H) for all analysis henceforth, given its ease of implementation and its superiority in convergence, as illustrated above. For each computation, we shall train the network till self-convergence in the estimated energy is attained starting from random initial parameters. For the examples illustrated in this report this is typically achieved $\le 100$ training epochs. To account for statistical fluctuations in self-converged energy, we shall henceforth always plot the mean energy reported for the last 20 iterations from our algorithm and report the corresponding standard deviation within error bars. Such energies will be denoted henceforth as $\rm{E_{RBM}}$ and the corresponding plot as simply $\rm{RBM}$. In the absence of self-convergence within a preset error threshold, we also warm-start from the converged parameter set of a similar, closely related problem. This typically accelerates not only the rate of convergence but is also known to enhance accuracy \cite{sajjan2021quantum, sureshbabu2021implementation}. At each training epoch, we shall use 10,000 samples with an initial burn-in of 10\% of the total samples for estimating average energy and gradients. The efficacy of the learning technique will be adjudicated by not only investigating the deviation of the finally converged energy with respect to exact diagonalization ($\rm{E}_{ED}$) but also by investigating the variance of the energy to ensure proximity to a true energy eigenstate. It must be emphasized that in MCMC-enabled schemes like ours, computation of the energy variance is inexpensive and obviates raising the Hamiltonian to higher powers. The variance of energy can be simply estimated through sampling configurations as 
\begin{eqnarray}
\sigma_H(\vec{X}) &\approx& \rm{Var}(H) \nonumber\\ &=& \frac{\sum_{\vec{v} \sim \phi} \kappa(\vec{v},\vec{X})(E_{loc}(\vec{v},\vec{X})-\mu_{\langle H \rangle(\vec{X})})^2}{\sum_{\vec{v} \sim \phi} \kappa(\vec{v}, \vec{X})} 
\label{eq:sample_var}
\end{eqnarray}
where $\rm{Var}(H)$ is the estimated sample variance and $\sigma_H(\vec{X})$ is the corresponding true population variance. Since $(E_{loc}(\vec{v},\vec{X}),\mu_{\langle H \rangle(\vec{X})},  \kappa(\vec{v}, \vec{X}))$ are all necessarily computed in the algorithm at every step, estimation of Eq.\ref{eq:sample_var} is naturally hassle-free. Besides being a diagnostic tool for the quality of the final converged state, computation of the sample variance serves a different purpose, too, wherein it can be used to even improve the final converged energy through zero-variance extrapolation (ZVE) \cite{fu2024variance,medvidovic2024neural}. The protocol hinges on the requirement that the true ground state will have zero variance and hence relies on the y-intercept of the $(\frac{\rm{Var(H)}}{\rm{E}_{\rm{RBM}}^2},\rm{E}_{\rm{RBM}})$ plot as the true and improved measure of the converged energy. As seen in the data illustrated below, and as has been noted before\cite{robledo2024chemistry,kaliakin2024accurate}, near the end of the training protocol such a $\frac{\rm{Var(H)}}{\rm{E}_{\rm{RBM}}^2}$ vs $\rm{E}_{\rm{RBM}}$ plot would exhibit a linear trend thereby rendering simple linear extrapolation techniques sufficient. We shall refer to estimated energies refined in this process as $\rm{E}_{\rm{RBM+ZVE}}$ and the corresponding plots as $\rm{RBM+ZVE}$. It must be emphasized that ZVE can only provide corrections to the estimated energy and not to the parameters of the converged state. To judge the quality of the converged state, we shall use the parameter set corresponding to the minimum converged energy within the last 20 iterations to compute subsequent properties. 

\begin{figure}[ht!]
    \centering
    \includegraphics[width=1.0\linewidth]{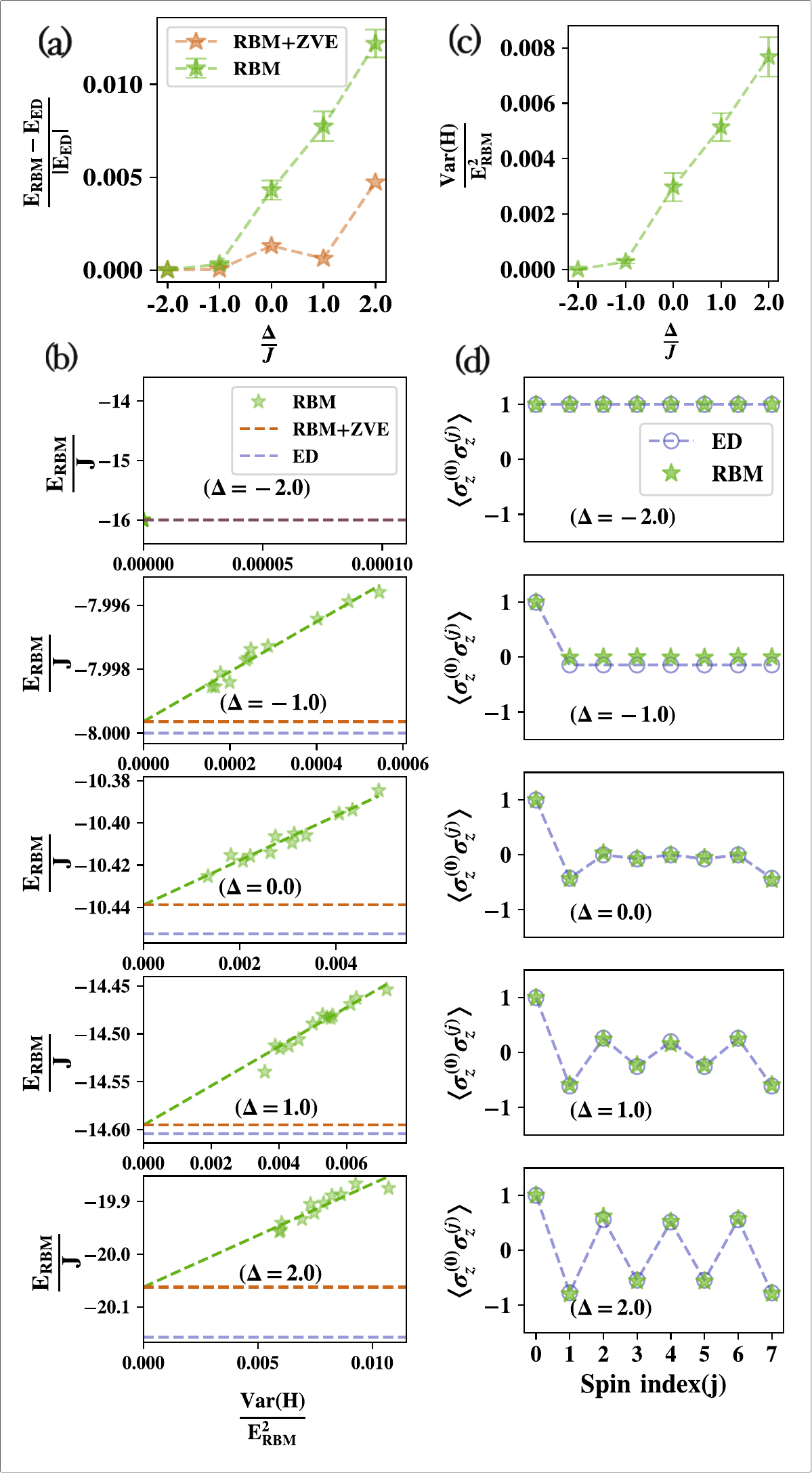}
    \caption{The relative errors of the estimated ground-state energy $\rm{E}_{\rm{RBM}}$ versus the exact energy $\rm{E}_{\rm{ED}}$ for the XXZ model are analyzed for $\frac{\Delta}{J} \in (-2.0, -1.0, 0, 1.0, 2.0)$. (a) The green curve ($\rm{RBM}$) shows the mean energy from the last 20 iterations, with error bars indicating respective standard deviations to account for fluctuations in self-convergence. The orange curve ($\rm{RBM+ZVE}$) demonstrates improved relative errors using zero-variance extrapolation (ZVE), reducing errors by a factor of 3 to within 0.5\%.(b) ZVE extrapolation is shown by plotting $\rm{E}_{\rm{RBM}}$ against $\frac{\rm{Var(H)}}{\rm{E}_{\rm{RBM}}^2}$, with extrapolated energies (orange lines) compared to exact values (blue dashed lines). The extrapolation scheme is shown for each $\frac{\Delta}{J} \in (-2.0, -1.0, 0, 1.0, 2.0)$ (c) The average $\frac{\rm{Var(H)}}{\rm{E}_{\rm{RBM}}^2}$ over the last 20 iterations is presented for the same $\frac{\Delta}{J}$.(d) The static two-point correlation functions $(\langle \sigma_z^{(0)}\sigma_z^{j}\rangle)$ from RBM match the exact state across phases: antiferromagnetic ($\frac{\Delta}{J} \ge 1.0$), XY ($-1.0 \leq \frac{\Delta}{J} \leq 1.0$), and ferromagnetic ($\frac{\Delta}{J} \le -1.0$).}

    \label{fig:3rd_panel}
\end{figure}

The first example we shall treat corresponds to the celebrated spin-$\frac{1}{2}$-Heisenberg XXZ model as the driver. The model is important for understanding quantum magnetism due to the first justification of exchange interaction and its subsequent analytical solvability using Bethe ansatz \cite{bethe1931theorie,batchelor2007bethe}. It has undergone renewed interest in recent years due to its recent experimental\cite{scheie2021detection} and theoretical\cite{ljubotina2019kardar,dupont2021spatiotemporal,PhysRevB.101.045115,ye2022universal} reports claiming that it can support anomalous superdiffusive spin transport which can have unprecedented consequences in spintronics\cite{wolf2001spintronics,guo2019recent,garcia2018spin}. In its 2D variant, it is also extensively used in modelling ground states of $\mathcal{Z}_2$ spin liquids like in minerals like Herbertsmithite \cite{kagomel}. Besides, it is used as a benchmark for the implementation of quantum computing algorithms \cite{PhysRevLett.131.060406,PhysRevResearch.5.013183}, cold atom based simulation \cite{jepsen2021transverse,PhysRevLett.125.240504}, trapped-ion based simulation \cite{PhysRevLett.92.207901,RevModPhys.93.025001} or studies connected to phase transitions\cite{chen2007fidelity} and thermalization \cite{rigol2013dynamics}. It also acts as a platform for the active development of condensed matter physical tools for comparing with higher spin analogues like spin-$1$ systems wherein topological Haldane phases can be seen with exponentially decaying correlation functions of local observable and an excitation gap\cite{PhysRevB.100.144423}. In fact, several variants of the model have now been well studied, including the inclusion of integrability terms within like perturbation with staggered field\cite{huang2013scaling,vasseur2016nonequilibrium} or disorder\cite{avishai2002level} which can simulate real materials under certain circumstances.

\begin{figure*}[ht!]
    \centering
    \includegraphics[width=0.99\textwidth]{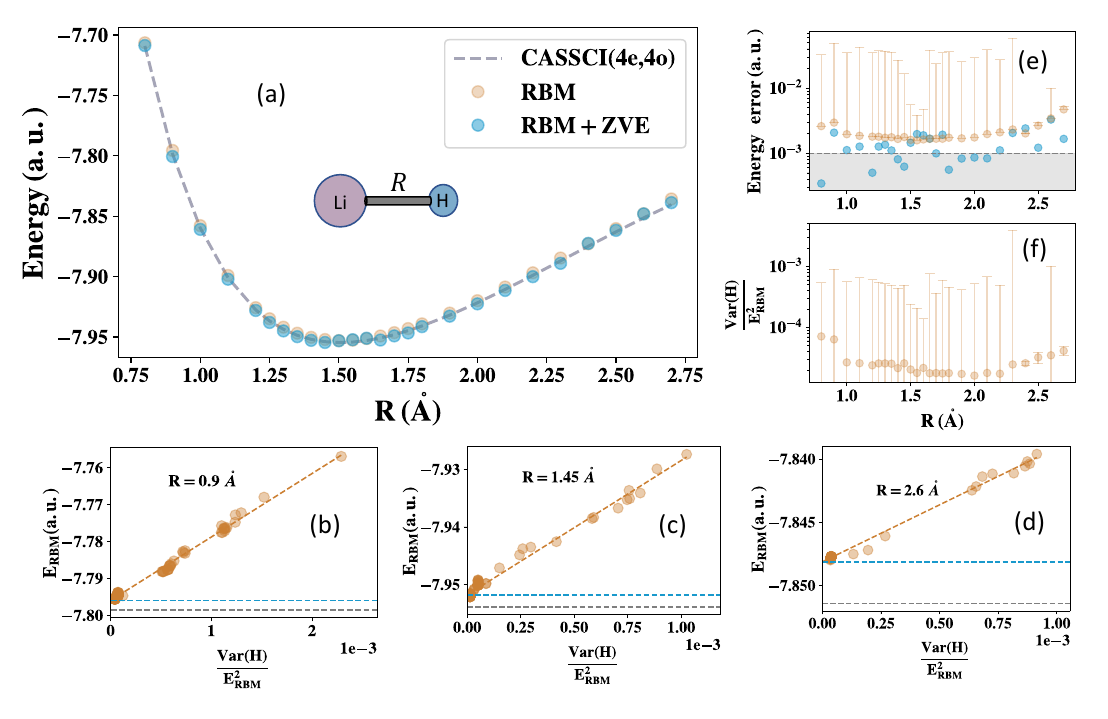}
    \caption{(a) The ground-state energy (a.u.) vs.bond length ($\rm{R} \:(\mathring{A})$) of $\:\rm{Li-H}$ is shown. For computations without ZVE (orange), the minimum energy value from the last few data points in the training protocol is used. For RBM + ZVE (blue), the y-intercept from the extrapolation scheme (as described in the text) is used. Both methods, especially RBM + ZVE, show good agreement with the exact CASSCI results (dashed gray line) in an active space of (4e,4o) across different bond length .(b-d) The RBM-ZVE procedure is illustrated for three bond lengths from distinct regions of the surface: (b) $\rm{R} = 0.9 \:\mathring{A}$, (c) $\rm{R} = 1.45 \:\mathring{A}$, and (d) $\rm{R} = 2.6 \:\mathring{A}$. In each case, the final extrapolated energy (blue horizontal dashed line) is compared to the exact CASSCI energy (gray dashed line). (e) The absolute energy errors (a.u.) relative to CASSCI are shown for both methods. RBM-ZVE (blue) achieves errors at/below the chemical threshold ($\le 10^{-3}$ a.u., shaded in gray) for a considerable number of bond lengths. Error bars represent the standard deviation of energy errors over the last few points used in RBM-ZVE. (f) The variance $\frac{\rm{Var(H)}}{\rm{E}_{\rm{RBM}}^2}$ vs. $\rm{R} \:(\mathring{A})$ is plotted. For each bond length, parameters corresponding to the minimum energy among the last few points are used, with error bars showing the standard deviation of $\frac{\rm{Var(H)}}{\rm{E}_{\rm{RBM}}^2}$. The nearly vanishing variance confirms that the final state approximates the ground stationary state well.}
    \label{fig:4th_panel}
\end{figure*}

The Hamiltonian for the driver is 
\begin{eqnarray}
H = \sum_{\langle i,j \rangle}^N J(\sigma_x^{(i)}\sigma_x^{(j)} + \sigma_y^{(i)}\sigma_y^{(j)}) + \Delta \sigma_z^{(i)}\sigma_z^{(j)}
\label{eq:XXZ_driver}
\end{eqnarray}
where $<i,j>$ indicates nearest-neighbor interaction and $\sigma_{\alpha}^{({i})}$ denotes a Pauli spin operator ($\alpha \in \{x,y,z\}$) at $i$-th spin and $(\Delta, J) \in \mathbb{R}^2$. It is well-known that the ground state of the Hamiltonian in Eq.\ref{eq:XXZ_driver} is anti-ferromagnetic when $\frac{\Delta}{J} \ge 1$, ferromagnetic when $\frac{\Delta}{J} \le -1$, and in an XY phase with vanishing local order \cite{PhysRevB.93.054417} when $-1 \le \frac{\Delta}{J} \le 1$. At $\frac{\Delta}{J}=-1$, the phase transition between XY phase and the ferromagnet is first-order, whereas at $\frac{\Delta}{J}=1$, the phase transition between XY phase and the anti-ferromagnet is a Kosterlitz-Thouless infinite-order one \cite{PhysRevA.85.052128}. We shall reproduce the ground state energy and properties of the system in all three regimes using this newly developed algorithm. For computations, we use $N=8$ spins and set $J=1$. The values of $\Delta$ we choose as representative examples for various phases are $[-2.0, -1.0, 0.0, 1.0, 2.0]$. In Fig.\ref{fig:3rd_panel}(a), we plot the relative energy error for our algorithm using exact diagonalization as reference labeled as RBM in green. For each $\Delta$ the relative energy error plotted is averaged over last 20 iterations to account for self-convergence with the errorbar being the corresponding standard deviation. We see that when RBM+ZVE is applied, the relative errors decrease nearly by 2.5 times, as seen in orange. In such a case, the relative errors for each of the phases are $<5\times10^{-3}$. The procedure for ZVE is illustrated in Fig.\ref{fig:3rd_panel}(b) for each of the $\Delta$ values. In the respective panels, the exact energy from ED is depicted with blue dashed lines, and the extrapolated RBM+ZVE energy is denoted in orange to show its deviation from ED. In Fig.\ref{fig:3rd_panel}(c) we plot the averaged 
$\frac{\rm{Var(H)}}{\rm{E}_{\rm{RBM}}^2}$ from the respective panels in Fig.\ref{fig:3rd_panel}(b) with the standard deviation as the errorbar. In Fig.\ref{fig:3rd_panel}(d) we plot static two-point spin correlation function $\langle \sigma_z^{(0)}\sigma_z^{j}\rangle$ of the first spin (labelled as index '0') and the $j$-th spin along the lattice for various values of $\Delta$ using the converged parameter set corresponding to minimum reported energy in the last 20 iterations plotted in Fig.\ref{fig:3rd_panel}(a). We see that for each $\Delta$, the two-point correlation functions are appropriately reproduced in  Fig.\ref{fig:3rd_panel}(d) when compared with the ground state from exact diagonalization. 

\begin{figure*}[ht!]
    \centering
    \includegraphics[width=0.99\textwidth]{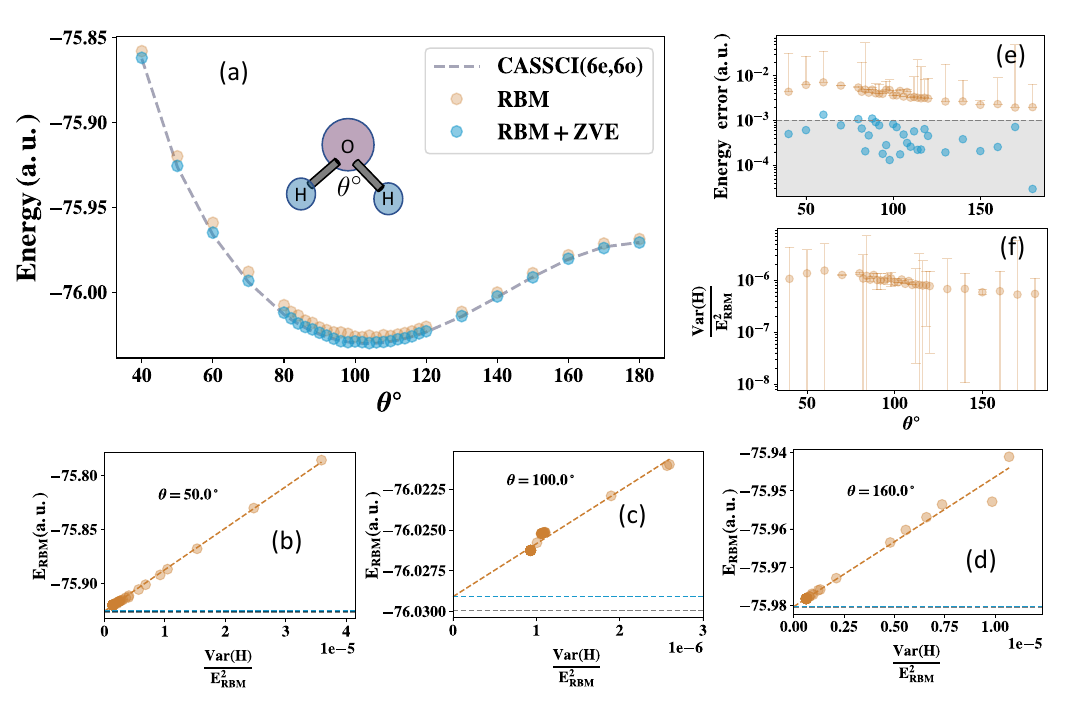}
    \caption{(a) The ground-state energy (a.u.) vs. $\angle \:\:\rm{H-O-H}$ ($\theta^{\circ}$) is shown. For computations without ZVE (orange), the minimum energy value from the last few data points in the training protocol is used. For RBM + ZVE (blue), the y-intercept from the extrapolation scheme (as described in the text) is used. Both methods, especially RBM + ZVE, show good agreement with the exact CASSCI results (dashed gray line) in an active space of (6e,6o) across bond angles.(b-d) The RBM-ZVE procedure is illustrated for three bond angles from distinct regions of the surface: (b) $\theta=50^\circ$, (c) $\theta=100^\circ$, and (d) $\theta=160^\circ$. In each case, the final extrapolated energy (blue horizontal dashed line) is compared to the exact CASSCI energy (gray dashed line). (e) The absolute energy errors (a.u.) relative to CASSCI are shown for both methods. RBM-ZVE (blue) consistently achieves errors at/below the chemical threshold ($\le 10^{-3}$ a.u., shaded in gray) for most bond angles, including those far from equilibrium geometry. Error bars represent the standard deviation of energy errors over the last few points used in RBM-ZVE. (f) The variance $\frac{\rm{Var(H)}}{\rm{E}_{\rm{RBM}}^2}$ vs. $\angle \:\:\rm{H-O-H}$ ($\theta^{\circ}$) is plotted. For each angle, parameters corresponding to the minimum energy among the last few points are used, with error bars showing the standard deviation of $\frac{\rm{Var(H)}}{\rm{E}_{\rm{RBM}}^2}$. The nearly vanishing variance confirms that the final state approximates the ground stationary state well.}
    \label{fig:5th_panel}
\end{figure*}

The next set of examples that we shall study will be the modeling of electronic ground states of molecular systems. The usual framework for studying such systems is to invoke the Born-Oppenheimer approximation, where heavier atomic nuclei are considered frozen in space as classical point charges, and only the eigenvalue problem of electronic Hamiltonian is solved for such a fixed set of generalized nuclear coordinates (say denoted by $\vec{R}$). Using basis functions in real space following the Ritz-Galerkin discretization scheme, such an electronic Hamiltonian can be expressed as 
\begin{equation}
    H = \sum_{i,k}^{N_o} h_{ik}(\vec{R}) a_{i}^{\dagger} a_k + \sum_{ijkl}^{N_o} V_{ijkl}(\vec{R}) a_{i}^{\dagger} a_{j}^{\dagger} a_k a_l 
    \label{eq:elec_Ham}
\end{equation}
where $N_o$ is the single-particle spin-orbital basis rank and usually scales linearly with the number of electrons in the system. The terms $h_{ik}(\vec{R})$ and $V_{ijkl}(\vec{R})$ are the usual one and two-body electronic integrals expressed in the chosen basis and $a_{i}^{\dagger}$ being the usual creation operator responsible for introducing excitation in the $i$-th fermionic mode/spin orbital and $a_i$ being its hermitian adjoint \cite{helgaker2013molecular}. Jordan-Wigner (JW) transformation of the operators $(a_{i}^{\dagger},a_i)$ in terms of Pauli words is initiated to respect their anti-commutation\cite{bauer2020quantum}. This converts the electronic hamiltonian in Eq.\ref{eq:elec_Ham} into a non-local spin hamiltonian of the form $H=\sum_s c_s(h_{ik}(\vec{R}), V_{ijkl}(\vec{R})) P_s$ where $P_s = \otimes_{q=1}^n \sigma_\alpha^{(q)}, \: \alpha \in \{x,y,z,0\}$ and $c_s(h_{ik}(\vec{R}), V_{ijkl}(\vec{R})) \in \mathbb{R}$. This is the starting point of our algorithm. Number of such Pauli words is usually $O(N_{o}^4)$ but can be significantly reduced to $O(N_{o}^p)$ where $2\le p \le 3$ by tensor factorizations or use of basis functions with vanishing overlap\cite{PRXQuantum.2.030305, PhysRevX.8.011044}.

For benchmarking computational accuracy offered by our algorithm, we shall choose two molecular examples - the first being bond-stretching of a prototypical example like $\rm{LiH}$ and the second being angular distortion of the $\angle \:\:\rm{H-O-H}$ in $\rm{H_2O}$. We shall use $\rm{ccPVDz}$ basis set for the atomic orbitals for both systems. After an initial Hartree-Fock calculation to define the respective molecular orbitals (MOs), we used a CASSCI Hamiltonian. For
$\rm{H_2O}$ we use an active space of (6e,6o) (12 spin orbitals) whereas for $\rm{LiH}$ we use an active space of (4e,4o) (8 spin orbitals) . These active spaces are chosen from the frontier set of MOs to allow energy-efficient excitations beyond the mean-field determinantal wavefunction. We saw that these active spaces resulted in a smooth potential energy surface (PES) with all the principal features accurately captured when compared against a bigger active space. Since the Hamiltonian generated in these active spaces is subsequently JW transformed and fed as input into our algorithm, energy from the exact diagonalization of this reduced Hamiltonian (CASSCI) in the chosen active space will be used as ground truth for benchmarking.  

In Fig.\ref{fig:4th_panel}(a) we plot the results of the potential energy surface (PES) of the ground-state energy (a.u.) vs. length of the $\rm{Li-H}$ bond ( $\rm{R} \:(\mathring{A})$). Similar plot is displayed in Fig.\ref{fig:5th_panel}(a) for the potential energy surface (PES) of the ground-state energy (a.u.) vs. $\angle \:\:\rm{H-O-H}$ ($\theta^{\circ}$) in $\rm{H_2O}$. In either case, results from both flavors of computation i.e. usual RBM (orange in both plots) and RBM+ZVE (blue in both plots) are shown and when compared against the CASSCI results (gray dashed line in both cases) in the chosen active space shows considerable agreement, especially RBM + ZVE. In Fig.\ref{fig:4th_panel} (or Fig.\ref{fig:5th_panel}) (b-d) we illustrate the procedure for zero-variance extrapolation at three different bond lengths (bond angles) for $\rm{LiH}$ ($\rm{H_2O}$) across three different regimes of the PES drawn in Fig.\ref{fig:4th_panel} (Fig.\ref{fig:5th_panel})(a) using few points towards the end of the training process. We see that the extrapolated energy (blue dashed horizontal line in both figures) is always closer to the exact energies (gray dashed horizontal line in both figures) than those at the individual iterations explaining why ZVE produces more accuracy we see in the subfigure (a) for respective panels of Fig.\ref{fig:4th_panel} or Fig.\ref{fig:5th_panel}. This is further bolstered in Fig.\ref{fig:4th_panel} (Fig.\ref{fig:5th_panel})(e) where the energy errors (a.u.) with respect to CASSCI are reported. It is clearly seen in both the plots for the respective systems that for a considerable number of parameter values, RBM+ZVE produces energy errors below the chemical accuracy threshold (shaded in gray) even with conservative resource requirements (see Section IV). The points in orange in Fig.\ref{fig:4th_panel}(Fig.\ref{fig:5th_panel}) (e), denote the energy error from the minimum energy obtained over the last few iterations as used in (b-d) for each figure, with the error bar denoting the respective fluctuation. In Fig.\ref{fig:4th_panel} ( Fig.\ref{fig:5th_panel})(f) we plot the $\frac{\rm{Var(H)}}{\rm{E}_{\rm{RBM}}^2}$ from the last few points in (b-d) of respective panels. The low magnitude of this quantity in either case shows that the RBM training has produced a final state which is a good approximant to the energy eigenstate.


In this work, we present an algorithm for training neural network representations of quantum states using quantum-assisted methods. The algorithm efficiently approximates a target stationary state given a system Hamiltonian. We provide a systematic, polynomially efficient mapping of energy-based generative models, such as Restricted Boltzmann Machines (RBMs), to surrogate networks. By analytically demonstrating that the surrogate can model any arbitrary probability density, the approach maintains generality and is extendable to other neural quantum states (NQS) for $n$-qubit systems. Training exploits the surrogate network’s distribution via variational Monte Carlo to estimate energy and parameter gradients. Using a quantum circuit to generate bit strings distributed according to the surrogate state improves mixing times, reduces burn-in, and enhances the quality of sampled distributions. The quantum circuit, implemented with a Trotterization scheme, operates efficiently with $O(n)$ and $O(\tau n)$ gate depths and requires only $O(N_s)$ queries, ensuring scalability. Our method captures both amplitude and phase of the target state while requiring only $O(mn)$ storage, where $m = O(\alpha n)$ and $\alpha \in \mathbb{R}$, surpassing previous approaches that required full distribution storage and additional classical preprocessing for the phase. We validate the algorithm through ground-state learning of diverse systems, including local spin models and non-local electronic Hamiltonians, achieving high accuracy. With zero-variance extrapolation, implemented at no additional cost, the method achieves even greater precision, highlighting its versatility. 

Quantum state-learning is a critical focus in many-body physics, aiming to efficiently represent quantum states within exponentially large Hilbert spaces. Over decades, various classical approximate methods have emerged, particularly in chemistry, where wave function-based approaches like coupled-cluster (CC) theory\cite{bartlett2007coupled} and perturbative methods (e.g., MP2\cite{cremer2011moller}) have been widely applied to capture dynamic correlation. However, their reliance on a single-determinant ground state makes them inadequate for treating strong multireference/static electronic correlation seen in distorted molecular geometries, extended $\pi$-systems, reaction intermediates, and transition metal complexes. Methods like CASSCF \cite{olsen2011casscf}, which include all determinantal excitations within an active space, are often used for such cases. Its run-time scales as ${r \choose N_{\alpha}} \times {r \choose N_{\beta}}$, where $r = r_o + r_{uo}$ is the single-particle basis rank, and $r_o$($r_{uo}$) are the occupied(unoccupied) orbitals in the Hartree-Fock reference and $N_{\alpha}(N_{\beta})$ is the number of spin-up(down) electrons in the system. Such a factorial asymptotic complexity makes it computationally prohibitive for extension to larger systems. For example, even a modest active space of 18 orbitals and 18 electrons involves handling $O(10^9)$ determinants. Similarly, even CC methods, although comparatively cheaper, can still be computationally demanding, with costs scaling as $O(r_o^2 r_{uo}^4)$ for CCSD and $O(r_o^3 r_{uo}^4)$ for CCSD(T)\cite{helgaker2000coupled}. Similarly, for quantum spin models, tensor network methods\cite{orus2019tensor} like DMRG have proven to be an excellent choice for Hamiltonians with short-range interactions, particularly in 1D systems even though innovative tensor trains in higher dimensions exist too\cite{verstraete2004renormalization}. However, their efficiency diminishes in higher dimensions or systems with significant non-local quantum correlations, as the increasing number of non-trivial singular values necessitates severe truncation, compromising accuracy. 

Our method can also be compared to other quantum-enabled approaches for quantum state learning. Over the past few decades, a suite of quantum-enabled algorithms have been designed for a variety of tasks \cite{biron1999generalized, ekert1996quantum, montanaro2016quantum, selvarajan2023dimensionality, sajjan2023physics,kumaran2024random,gupta2022hamiltonian}. Similarly, even for quantum state learning, robust protocols such as quantum phase estimation (QPE) \cite{abrams1999quantum,aspuru2005simulated} and quantum singular value transformation (QSVT) \cite{gilyen2019quantum,dong2022ground} have been developed to approximate energy eigenpairs. However, these methods require deeper circuits, which expand the light cone of measured observables and reduce fidelity due to hardware noise in modern quantum devices \cite{de2021materials}. To address these issues and design algorithms suitable for near-term applications, significant effort has focused on leveraging the variational principle through parameterized low-depth unitary circuits, often referred to as ansatz. For chemical problems, common ansatz choices include unitary coupled cluster ansatz \cite{anand2022quantum} (UCCSD), hardware-efficient ansatz (HEA) \cite{kandala2017hardware}, and qubit-coupled cluster ansatz \cite{ryabinkin2018qubit}, symmetry adapted ansatz design\cite{gard2020efficient}. Despite being motivated physically, these approaches often demand prohibitive gate depths for reasonable accuracy. For instance, pristine UCCSD requires $O(r^4-r^5)$ gates and $O(r_o^2 r_{uo}^2)$ cluster amplitudes as parameters. Dynamic operator selection strategies, such as ADAPT-VQE \cite{grimsley2019adaptive} or its orbital-overlap adapted \cite{feniou2023overlap} routine, mitigate some of these issues by choosing operators based on their contributions to the energy gradient. However, these approaches often involve increased measurements and their expressibility is heavily limited by the preset operator pool which in turn, impacts gate count, algorithmic efficiency, and simulation accuracy. Additionally, systematic studies have highlighted the emergence of barren plateaus \cite{cerezo2021variational,larocca2024review}, raising concerns about the trainability of such models

Our algorithm, rooted in variational Monte Carlo, eschews many of the above-mentioned problems
and operates differently by using the circuit as a sample generator, thereby offering clear advantages.
Only $O(mn)=O(\alpha n^2)$ parameter gradients and a single energy term must be handled (for electronic problems $n=2r$ if a simple Jordan-Wigner mapping is used\cite{weidman2024quantum}) with a constant number of queries to the circuit defined over $O(n)$ gate depth and qubits. Besides, the mixing time for quantum-enhanced sampling is seen to be reduced by a factor of 3 compared to classical samplers, and unlike other schemes, direct computation of the transition probability matrix element is not necessary due to symmetry. Classically for electronic problems, variational Monte Carlo schemes in the past have mainly been done with a determinantal basis endowed with a Jastrow prefactor as correlator \cite{sorella2005wave,toulouse2007optimization}. Even though such schemes do not encounter fermionic sign problems, unlike other flavors of MC-based algorithms, \cite{toulouse2016introduction} their accuracy in the past was often limited by the expressibility of the ansatz used. However, with the architectural diversity of neural networks and provable guarantees of their expressibility, as quantum states \cite{sharir2022neural, medvidovic2024neural}, the flexibility to design an adequate representation is now endless. Unlike variational methods designed with parameterized quantum circuits, training in our algorithm also uses well-defined analytical gradients (see Section E in Appendix), avoiding the need for param-shift rule-based estimations. This efficiency stems from the knowledge of the functional forms of the original and surrogate network distributions, a common feature for most neural quantum states (NQS). Future extensions of our method could explore dynamical evolution in topological materials \cite{bao2022light} or excited states in complex systems using recently analyzed VS scores \cite{wu2024variational}, such as in polaritonic chemistry \cite{ruggenthaler2023understanding} or in dynamical evolution of open quantum systems \cite{shivpuje2024designing}. For strongly correlated systems, incorporating symmetry and orbital rotations represents an unexplored direction and can be used to study molecular systems with controllable external field \cite{sajjan2018entangling}. Combined with advanced architectures like RNNs or transformers, these improvements could enable large-scale quantum-assisted computations, highlighting tantalizing prospects for the future of NQS.

\section*{Acknowledgments}
We acknowledge funding from the US Department of Energy, the Office of Science through the Quantum Science Center (QSC), a National Quantum Information Science Research Center.

\newpage
%


\widetext
\pagebreak
\appendix
\AtAppendix{\counterwithin{lemma}{section}}
\AtAppendix{\counterwithin{theorem}{section}}
\AtAppendix{\counterwithin{figure}{section}}



\def\mathunderline#1#2{\color{#1}\underline{{\color{black}#2}}\color{black}}

\section{Definitions/Notations}
We re-iterate the definitions considered in the main text and to be used subsequently for remaining parts of this Appendix.

\begin{itemize}
\item \(\vec{v} = (v_1, v_2, \dots, v_n)\): A binary configuration of the \(n\) visible units, where each \(v_i\) is a binary variable taking values in \(\{-1, 1\}\). Each configuration represents a possible classical state of the graph $G_2$ or the visible set of neurons of graph $G_1$. Unless otherwise specified, we shall sometimes omit the vector sign above and simply use \(v \in [0, \cdots, 2^n - 1]\). This is akin to representing the corresponding integer index for the configuration in binary (with -1 substituted by 0).

\item $\mathbb{V}=\{\vec{v}| \vec{v} \in \{-1,1\}^n\}$. One must note $|\mathbb{V}|=2^n$. 

\item \(\rho^v_{v'}\): The matrix element of the state \(\rho(\vec{X})\) corresponding to basis states \(v\) and \(v'\). It represents the density matrix of the state. 
\begin{align}
\rho^v_{v^\prime}(\vec{X}) \propto exp \left({-\beta \sum_i a_i v_i + \sum_i a_i^* v_i^\prime}\right)  \:\prod_{j=1}^m \Gamma_j (\vec{v},\vec{b}, \vec{W}) \:\Gamma_j (\vec{v}^\prime,\vec{b^*}, \vec{W^*}) 
\end{align} 
where, \(\Gamma_j(\vec{v}, \vec{b}, \vec{W}) = \cosh \left[ \beta \left(b_j + \sum_{i=1}^n W_{ij} v_i\right) \right]\) and
\begin{itemize}
    \item \(a_i\): A complex-valued parameter associated with the bias for the \(i\)-th visible unit in the system.
    \item \(b_j\): A complex-valued parameter associated with the bias for the \(j\)-th hidden unit in the system. 
    \item \(W_{ij}\): The weight (coupling) between the \(i\)-th visible unit and \(j\)-th hidden unit.
\end{itemize}

\item \(\vec{X}\): A shorthand representation for all the parameters \((\vec{a}, \vec{b}, \vec{W})\) that define the state \(\rho(\vec{X})\). These include complex-valued parameters associated with interactions and biases in the network.

\item \(\beta\): The inverse temperature parameter. 

\item \(H_v^{v'}\): The Hamiltonian matrix element of the driver (system being studied) between state configurations \(v\) and \(v'\).

\item The diagonal elements, \(\rho^v_v\), are real and satisfy the normalization condition \(\sum_v \rho^v_v = 1\). They can be seen as a probability distribution. 
\begin{eqnarray}
\rho^v_{v}(\vec{x}) &\propto& exp \left({-2\beta \sum_i \rm{Re}(a_i) v_i}\right)  \:\prod_{j=1}^m \left|\cosh\left(b_j + \sum_{i=1}^n W_{ij} v_i\right) \right|^2 
\end{eqnarray}

\item \(E_{\mathrm{loc}}(v) = \frac{\sum_{v'} \rho_{v'}^v H_v^{v'}}{\rho^v_v}\): The local energy of the driver for the configuration \(v\). It measures the effective energy associated with the state \(v\), computed by taking into account all the other configurations $v'$ to which it is connected via the Hamiltonian matrix element of the driver.

\item \(\phi(v)\): The distribution that represents the state of the surrogate network. It is defined as 
\begin{eqnarray}
\phi_{\vec{v}}(\vec{l}(\vec{X}), \vec{J}(\vec{X})) &\propto& e^{-\beta \sum_i l_i(\vec{X}) v_i + \sum_{ij} J_{ij}(\vec{X}) v_i v_j} 
\end{eqnarray}
Here $l_i \in (\vec{l}(\vec{X})$ denotes the on-site field at each spin of the surrogate network $G_2$ in main text and  $J_{ij}(\vec{X}) \in \vec{J}(\vec{X}))$ denotes the mutual coupling between a pair of spins of $G_2$. $\phi(v)$ satisfies the normalization condition $\sum_v \phi(v) = 1$.

\item \(\kappa(v)\): A configuration dependent pre-factor chosen to establish the equivalence between \(\rho_v^v\) and \(\phi(v)\), as defined above
\[\rho_v^v \approx \kappa(v) \phi(v)\]


\end{itemize}

\newpage

\section{Proof of generality of factorization of probability distribution}

For a given probability distribution $P:\mathbb{V}\to [0,1]$, we wish to prove Theorem II.1 in the main text. But to do that we would first like to prove certain lemmas where the key quantities of interest will be :
\begin{enumerate}
    \item $F_{\vec{a}}(\vec{v}) = v_1^{a_1} v_2^{a_2} \cdots v_n^{a_n} = \prod_i v_i^{a_i} \:\:\: \vec{a}=(a_1, a_2....a_n)\:\: with\:\: a_i \in \{0,1\} \:\: and \:\: \vec{v}\:\:\in \mathbb{V}$. Each entry
 represents an element of a monomial binary basis.
    \item $\chi = \{F_{\vec{a}}(\vec{v})|\vec{a} \in \{0,1\}^n\}$
\end{enumerate}

We first would like to establish that the set $\chi$
can act as a basis for the space of functions $f:\mathbb{V} \to \mathbb{R}$ . We show this in two different ways using the following lemmas

\begin{lemma} (Approach 1: Orthogonality of Monomial basis)\label{lem:ortho_mono}
    Let us define two vectors $\vec{a}, \vec{b} \in \{0,1\}^n$ such that for the monomials we have 
    \begin{align*}
        F_{\vec{a}}(\vec{v}) &= v_1^{a_1} v_2^{a_2} \cdots v_n^{a_n} = \prod_i^n v_i^{a_i}\\
        F_{\vec{b}}(\vec{v}) &= v_1^{b_1} v_2^{b_2} \cdots v_n^{b_n} = \prod_i^n v_i^{b_i} 
    \end{align*}
    Then, with respect to the inner product $\langle F_{\vec{a}}(\vec{v}), F_{\vec{b}}(\vec{v}) \rangle = \frac{1}{2^n} \sum_{\vec{v} \in V} F_{\vec{a}}(\vec{v}) F_{\vec{b}}(\vec{v})$, the monomials $F_{\vec{a}}(\vec{v})$ and $F_{\vec{b}}(\vec{v})$ are orthonormal i.e. $\langle F_{\vec{a}}(\vec{v}), F_{\vec{b}}(\vec{v}) \rangle =0$ if $\vec{a} \ne \vec{b}$ and 1 otherwise . 
\end{lemma}

\begin{proof}
    Since $F_{\vec{a}}(\vec{v}):\mathbb{V}\to\mathbb{R}$, it is easy to show $\langle F_{\vec{a}}(\vec{v}), F_{\vec{b}}(\vec{v}) \rangle$ satisfies all three axioms of usual inner-products (linearity, positive semi-definiteness when $\vec{a}=\vec{b}$ and symmetry), thus we straightaway compute the inner product between two monomials $F_{\vec{a}}(\vec{v})$ and $F_{\vec{b}}(\vec{v})$:
    \begin{align*}
        \langle F_{\vec{a}}(\vec{v}), F_{\vec{b}}(\vec{v}) \rangle &= \frac{1}{2^n} \sum_{\vec{v} \in V} F_{\vec{a}}(\vec{v}) F_{\vec{b}}(\vec{v}) \\
        &= \frac{1}{2^n} \sum_{\vec{v} \in V} \prod_{i=1}^n v_i^{a_i \oplus_2 b_i} \\  
        &= \frac{1}{2^n} \sum_{\vec{v_1} \in \{1,-1\}}\sum_{\vec{v_2} \in \{1,-1\}}\cdot\cdot\sum_{\vec{v_n} \in \{1,-1\}} \prod_{i=1}^n v_i^{a_i \oplus_2 b_i} \\
        &= \left( \frac{1}{2} \sum_{v_1 \in \{1,-1\}}  v_1^{a_1 \oplus_2 b_1} \right)    \left( \frac{1}{2} \sum_{v_2 \in \{1,-1\}}  v_2^{a_2 \oplus_2 b_2} \right)    \cdots    \left( \frac{1}{2} \sum_{v_n \in \{1,-1\}}  v_n^{a_n \oplus_2 b_n} \right)\\
        &=\begin{cases}
        1 \:\:\:\:\:iff\:\:\:\:\:\: a_i = b_i \implies a_i \oplus b_i = 0 \:\:\forall i\in \{1,2,...n\} \\
        0 \:\:\:\: otherwise \:\:\: as\:\:\exists k \in \{1,2...n\}\:\: s.t  \:\:\: a_k \oplus b_k = 1 \implies \sum_{v_k \in \{1,-1\}} v_k = 0
    \end{cases}
    \end{align*}    

where `$\oplus_2$' is a modulo 2 sum, which behaves as an XOR operation since the variables $a_i$ and $b_i$ are binary. 
\end{proof}

\begin{lemma}(Approach 2 : Use of indicator functions)
\label{lem:indicator_formalism_LI}
Let us define an indicator function $I_{\vec{h}}(\vec{v}):\mathbb{V}\to \mathbb{R}$ as:
\begin{eqnarray}
    I_{\vec{h}}(\vec{v}) = \prod_k \frac{1}{2} \left( 1 + \frac{v_k}{h_k} \right) =
\begin{cases} 
 1 & \text{if } \vec{v} = \vec{h} \\
 0 & \text{if } \vec{v} \neq \vec{h} 
 \end{cases}
\end{eqnarray}
where $\vec{v}, \vec{h} \in \mathbb{V}$. The set \(\mathcal{S} = \{ I_{h^{(i)}}(\vec{v}) \mid h^{(i)} \in V \}\) act as a basis for the function space 
$f:\mathbb{V} \to \mathbb{R}$.
\end{lemma}

\begin{proof}
We establish the assertion above by using the following claims\\
\begin{enumerate}[align=left,leftmargin=*]
\item \textbf{Claim 1:} The set S defined above consists of linearly independent items i.e. 
$\sum_{i=0}^{2^n-1} c_i \left( I_{h^{(i)}}(\vec{v}) \right) = 0 \quad \forall \vec{v} \in \mathbb{V} 
\iff c_i = 0 \quad \forall i $
\begin{subproof}[Subproof] \label{lem:ind_lin_ind}
    To demonstrate this let us recall that 
    $|S|=2^n$ as it contains one 
     $I_{h^{(i)}}(\vec{v})$ for every $h^{(i)} \in \mathbb{V}$. The condition of linear independence will have to be true for every
     $\vec{v} \in \mathbb{V}$ which will actually generate a set of $2^n$ linear equations that will need to be solved. It is easy to show that the null space of the matrix kernel generating the set of linear equations contains only a trivial entry. For example if we investigate the veracity of the condition by substituting a specific configuration say $\vec{v}=\vec{v}^{(k)}$
     \begin{eqnarray}
         &\sum_{i=0}^{2^n-1}& c_i \left( I_{h^{(i)}}(\vec{v}^{(k)}) \right) = 0 \nonumber \\
         &\implies& c_0 \left( I_{h^{(0)}} (\vec{v}^{(k)})\right) +  c_1 \left( I_{h^{(1)}} (\vec{v}^{(k)})\right) +\cdot\cdot c_k \left( I_{h^{(k)}} (\vec{v}^{(k)})\right)+ \cdot \cdot c_{2^n-1} \left( I_{h^{(2^{n}-1)}}(\vec{v}^{(k)})\right) \nonumber = 0 \\
         &\implies& c_k \left( I_{h^{(k)}}(\vec{v}^{(k)})\right) = 0 \:\:\:\:\:\:\:\:\:\:\:\because\:\:\:  I_{h^{(k)}}(\vec{v}^{(j)}) =0 \:\: if \:\: j\ne k \nonumber \\
         &\implies& c_k = 0 \:\: \:\:\:\:\:\:\:\: \because\:\:\: I_{h^{(k)}}(\vec{v}^{(k)}) =1
     \end{eqnarray}
     Since no assumption is made on the nature/choice of the configuration $\vec{v}^{(k)}$, this can be proven to be true for any of the $2^n$ configurations in set $\mathbb{V}$. This means that $c_i=0 \:\:\forall\:\: i \in\:\:\{0,1,..2^{n-1}\}$ which establishes forward direction of the claim. The reverse direction i.e. if $c_i=0 \:\:\forall\:\: i$ then $\sum_{i=0}^{2^n-1} c_i \cdot \left( I_{h^{(i)}}(\vec{v}) \right) = 0$ is trivially true.
\end{subproof}
\item \textbf{Claim 2:}\label{lem:ind_form_span}
The set S defined above have enough degrees of freedom to span the space of functions 
$f:\mathbb{V} \to \mathbb{R}$
\begin{subproof}[Subproof]
    To demonstrate this we investigate a specific linear combination of the kind 
     $   g(\vec{v}) = \sum_{i=0}^{2^n-1} c_i \left( I_{h^{(i)}}(\vec{v}) \right)$.
Using Claim \ref{lem:ind_lin_ind} , we already know that $g(\vec{v})=0$ iff $c_i =0 \:\:\forall i$. This establishes the inability to express any $I_{h^{(i)}}(\vec{v})$ in terms of other elements of the set $S$ above and  makes the  number of independently tunable coefficients in $C=(c_0,c_1,....c^{2^n-1})$ to $2^n$, one for each of the unique entries $h^{(i)} \in \mathbb{V}$ where $|S|=2^n$. This should already provide expressible guarantees to any $f:\mathbb{V} \to \mathbb{R}$ given such a function can at most have $2^n$ independently tunable components for entry $\vec{v} \in \mathbb{V}$ and be completely specified by delineating these entries. 
 Let us now imagine that we would like to construct such a function $f(\vec{v})$ such that $f(\vec{v}^{(k)})=g(\vec{v}^{(k)})$ for every $\vec{v}^{(k)} \in \mathbb{V}$. In other words akin to Lagrange interpolation, $g(\vec{v})$ would serve as an interpolating polynomial for $f(\vec{v})$. 
\begin{eqnarray}
f(\vec{v}^{(k)}) \equiv g(\vec{v}^{(k)}) &=& \sum_{i=0}^{2^n-1} c_i \left( I_{h^{(i)}}(\vec{v}^{(k)}) \right)\nonumber \\
&=& c_k \left( I_{h^{(k)}}(\vec{v}^{(k)})\right) \:\:\:\:\:\:\:\:\:\:\:\because\:\:\:  I_{h^{(k)}}(\vec{v}^{(j)}) =0 \:\: if \:\: j\ne k \nonumber \\
f(\vec{v}^{(k)}) &=& c_k \:\: \:\:\:\:\:\:\:\:\:\:\: \because\:\:\: I_{h^{(k)}}(\vec{v}^{(k)}) =1
\end{eqnarray}
As before since no assumption is made on the nature/choice of the configuration $\vec{v}^{(k)}$, this can be proven to be true for any of the $2^n$ configurations in set $\mathbb{V}$ i.e. $f(\vec{v}^{(k)}) = c_k \:\:\:\: \forall k \:\:\in \{0,1,..2^{n-1}\} $. It must be also be emphasized that neither $f(\vec{v})$ nor $g(\vec{v})$ is defined at any other points other than $\vec{v} \in \mathbb{V}$ due to discreteness of their domains which makes $f(\vec{v}^{(k)}) = g(\vec{v}^{(k)})$ at all conceivable points. Also there is no restriction on the components $f(\vec{v}^{(k)})$ too and they can be arbitrarily set to create any function. Thus it is always possible to find a representation of an arbitrary function $f(\vec{v}) = \sum_{i=0}^{2^n-1} f(\vec{h}^{(i)}) \left( I_{h^{(i)}}(\vec{v}) \right)$
\end{subproof}
\end{enumerate}
Together Claim\ref{lem:ind_lin_ind} and Claim\ref{lem:ind_form_span} proves the assertion stated.
\end{proof}

\newpage
\begin{lemma}\label{lem:ind_basis_to_mono}
    For any arbitrary $f:\mathbb{V} \to \mathbb{R}$ one can represent $f(\vec{v}) = \sum_{\vec{a} \in \{0, 1\}^n} c_{\vec{a}} F_{\vec{a}}(\vec{v})$ where $c_{\vec{a}} \in \mathbb{R}$
\end{lemma}
\begin{proof}
    Given an arbitrary $f:\mathbb{V} \to \mathbb{R}$ we can do the following:
    \begin{eqnarray}
      f(\vec{v}) &=& \sum_{i=0}^{2^n-1} f(\vec{h}^{(i)}) \left( I_{h^{(i)}}(\vec{v}) \right) \:\:\:\:\:\:\:\because Lemma \:\:\ref{lem:indicator_formalism_LI}\nonumber \\
      &=& \sum_{i=0}^{2^n-1} f(\vec{h}^{(i)}) \prod_k^n \frac{1}{2} \left( 1 + \frac{v_k}{h_k^{(i)}} \right)\nonumber \\
      &=& \sum_{i=0}^{2^n-1} f(\vec{h}^{(i)}) \frac{1}{2^n}\sum_{\vec{a} \in \{0, 1\}^n} \prod_k^n\left(\frac{v_k}{h_k^{(i)}}\right) ^{a_k} \:\:\:\:\: \forall k \:\: a_k\in \{0,1\}, \:\: \vec{a}=(a_0,a_1...a_n)\nonumber \\
      &=& \sum_{i=0}^{2^n-1} f(\vec{h}^{(i)}) \frac{1}{2^n}\sum_{\vec{a} \in \{0, 1\}^n} \prod_k^n v_k^{a_k} \prod_k^n \left (h_k^{(i)}\right)^{-a_{k}} \:\:\:\:\: \because \:\: h_k^{(i)}\: \in \:\{1,-1\}\nonumber \\
      &=& \sum_{\vec{a} \in \{0, 1\}^n} \underbrace{\left(\frac{1}{2^n} \sum_{i=0}^{2^n-1}  f(\vec{h}^{(i)}) \prod_k^n \left (h_k^{(i)}\right)^{-a_{k}}\right)}_{=c_{\vec{a}}\:\in\:\: \mathbb{R} \:\:\:\:\because\:\:\: f(\vec{h}^{(i)})\:\: \in \:\:\mathbb{R}} \prod_k^n v_k^{a_k} \nonumber \\
      &=& \sum_{\vec{a} \in \{0, 1\}^n} c_{\vec{a}} \prod_k^n v_k^{a_k} \nonumber \\
      &=& \sum_{\vec{a} \in \{0, 1\}^n} c_{\vec{a}}  F_{\vec{a}}(\vec{v}) \:\:\:\:\:\:\:\:\vec{v}=(v_1,v_2,..v_n), \:\: F_{\vec{a}}(\vec{v}) =  \prod_i v_i^{a_i} \nonumber
    \end{eqnarray}
Alternatively we could also use the assertion of Lemma\: \ref{lem:ortho_mono}, to claim that the set $\chi$ forms a linearly independent set. There are $2^n$ monomials $F_{\vec{a}}(\vec{v})$ given each is characterized by a distinct $\vec{a} \in \{0,1\}^n$ so $|\chi|=2^n$. Since any arbitrary function $f:\mathbb{V} \to \mathbb{R}$ would possess $2^n$ independent components, the set of monomials $\chi$ can span the space as a basis.     
\end{proof}

\begin{lemma}\label{lem:P_exp_F(v)}
Any arbitrary probability distribution $P:\mathbb{V}\to [0,1]$ defined over a finite support of bit-strings $\vec{v} \in \mathbb{V}$ can be written in the form 
\begin{eqnarray}
P(\vec{v}) = \frac{e^{\sum_{\vec{a}\in\{0,1\}^n, \vec{a}\ne \vec{0}} C_{\vec{a}} F_{\vec{a}}(\vec{v})}}{\sum_{\vec{v}} e^{\sum_{\vec{a}\in\{0,1\}^n,\vec{a}\ne \vec{0}} C_{\vec{a}} F_{\vec{a}}(\vec{v})}}
\end{eqnarray}
where $F_{\vec{a}}(\vec{v}) = v_1^{a_1} v_2^{a_2} \cdots v_n^{a_n} = \prod_i v_i^{a_i} \:\:\: \vec{a}=(a_1, a_2....a_n)\:\: with\:\: a_i \in \{0,1\}$ and  $C_{\vec{a}} \in \mathbb{R}$.
\end{lemma}

\begin{proof}
Let us define a function $f(\vec{v})= \log(P(\vec{v}))$ where $P(\vec{v})$ is the discrete probability distribution of interest. Given $P:\mathbb{V}\to [0,1]$, we then have $\log(P):\mathbb{V} \to [-\infty,0]$. We can make the domain of $\log(P) \in \mathbb{R}$ if for every $v^{(k)} \in \mathbb{V}$ for which $P(\vec{v}^{(k)})= 0$, we can replace $P(\vec{v}^{(k)}) \approx \epsilon$ where 
$\epsilon \to 0_{+}$. This will ensure $\log(P):\mathbb{V}\to(-\infty,0] \subset \mathbb{R}$ for which the assertions of the above lemmas hold true. With this substitution one can do the following
\begin{eqnarray}
    f(\vec{v}) &=& \log(P(\vec{v})) \nonumber \\
               &=& \sum_{\vec{a} \in \{0, 1\}^n} C_{\vec{a}} F_{\vec{a}}(\vec{v}) \:\:\:\:\:\: \because \:\:Lemma \:\:\ref{lem:ind_basis_to_mono}\nonumber \\
               &=& C_0 +\sum_{\vec{a} \in \{0, 1\}^n, \vec{a}\ne \vec{0}} C_{\vec{a}} F_{\vec{a}}(\vec{v})\nonumber \\
            P(\vec{v}) &=& e^{C_0}e^{\sum_{\vec{a} \in \{0, 1\}^n, \vec{a}\ne \vec{0}} C_{\vec{a}} F_{\vec{a}}(\vec{v})} \label{eq:P_v_as_expo}
\end{eqnarray}
From Eq.\ref{eq:P_v_as_expo} it is clear that $P(\vec{v})= \frac{e^{\sum_{\vec{a}\in\{0,1\}^n, \vec{a}\ne \vec{0}} C_{\vec{a}} F_{\vec{a}}(\vec{v})}}{\sum_{\vec{v}} e^{\sum_{\vec{a}\in\{0,1\}^n,\vec{a}\ne \vec{0}} C_{\vec{a}} F_{\vec{a}}(\vec{v})}}$ where $e^{C_0}= \frac{1}{\sum_{\vec{v}}e^{\sum_{\vec{a} \in \{0, 1\}^n, \vec{a}\ne \vec{0}} C_{\vec{a}} F_{\vec{a}}(\vec{v})}}$ by normalization.
\end{proof}

\newpage

Using Lemma \ref{lem:P_exp_F(v)}, Lemma \ref{lem:ind_basis_to_mono}, Claim \ref{lem:ind_form_span} we are now in the position to prove the Theorem \ref{thm:factor_theo_main} in main manuscript which is reproduced here for completeness and continuity

\begin{theorem}\label{thm:factor_theo}
Given $\mathbb{V}=\{\vec{v}| \vec{v} \in \{-1,1\}^n\}$ where $n\in \mathbb{Z}_{+}$ and an arbitrary discrete probability distribution $P:\mathbb{V}\to [0,1]$, then the following statement is true
\begin{enumerate}
\item It is always possible to find a representation of $P(\vec{v})$ as 
\begin{eqnarray}
    P(\vec{v}) &=& \mathcal{N}\kappa(\vec{v}) \frac{e^{\sum_{\vec{a} \in \{0, 1\}^n, H(\vec{a})\le k} C_{\vec{a}} F_{\vec{a}}(\vec{v})}}{\sum_{\vec{v}}e^{\sum_{\vec{a} \in \{0, 1\}^n, H(\vec{a})\le k} C_{\vec{a}} F_{\vec{a}}(\vec{v})}} \nonumber \\
    &=& \mathcal{N}\kappa(\vec{v})\phi(\vec{v})
\end{eqnarray}
where $F_{\vec{a}}(\vec{v}) = \prod_i v_i^{a_i} \:\:\:, \vec{a}=(a_1, a_2....a_n) \in \{0,1\}^n ,\:\: \:\: \vec{v}=(v_1,v_2,...v_n)\:\:\in \mathbb{V}$, $\:\:\:C_{\vec{a}} \in \mathbb{R}$, and $H(\vec{a})=\sum_i^n a_i$ is the Hamming weight of $\vec{a}$. We have defined $\phi(\vec{v})=\frac{e^{\sum_{\vec{a} \in \{0, 1\}^n, H(\vec{a})\le k} C_{\vec{a}} F_{\vec{a}}(\vec{v})}}{\sum_{\vec{v}}e^{\sum_{\vec{a} \in \{0, 1\}^n, H(\vec{a})\le k} C_{\vec{a}} F_{\vec{a}}(\vec{v})}}$ as another discrete distribution $\phi:\mathbb{V}\to [0,1]$ which will be the backbone of a secondary surrogate network.
Also $\mathcal{N}\ne f(\vec{v}) \ge 0$ and $k\in \mathbb{Z}_{+}$, is a user-defined preset non-negative integer. $k$ controls the degree of the polynomial chosen for expressing $\phi(\vec{v})$.

\item $\kappa(\vec{v})$ is a configuration dependant non-negative prefactor. In the most generic case, $\exists$ $M \in \mathbb{R}\ge 0$ such that $|\log(\kappa(\vec{v}))|$ is upper bounded as  
\begin{eqnarray}
    |\log(\kappa(\vec{v}))| \le M\left(2^{n} -\sum_{j=0}^{k} \ ^nC_j \right)
\end{eqnarray}
\end{enumerate}
\end{theorem}

\begin{proof}
We establish each of the assertions as follows:\\
\begin{enumerate}[align=left,leftmargin=*]
\item \textbf{Statement 1:} $P(\vec{v}) = \mathcal{N}\kappa(\vec{v}) \frac{e^{\sum_{\vec{a} \in \{0, 1\}^n, H(\vec{a})\le k} C_{\vec{a}} F_{\vec{a}}(\vec{v})}}{\sum_{\vec{v}}e^{\sum_{\vec{a} \in \{0, 1\}^n, H(\vec{a})\le k} C_{\vec{a}} F_{\vec{a}}(\vec{v})}}=\mathcal{N}\kappa(\vec{v})\phi(\vec{v})$
\begin{subproof}[Subproof]
We can prove this assertion the following way
   \begin{eqnarray}
       P(\vec{v}) &=& \frac{e^{\sum_{\vec{a}\in\{0,1\}^n, \vec{a}\ne \vec{0}} C_{\vec{a}} F_{\vec{a}}(\vec{v})}}{\sum_{\vec{v}} e^{\sum_{\vec{a}\in\{0,1\}^n,\vec{a}\ne \vec{0}} C_{\vec{a}} F_{\vec{a}}(\vec{v})}} \:\:\:\:\:\:\:\:(See \:\:Lemma \:\ref{lem:P_exp_F(v)}) \nonumber \\
       &=& \frac{e^{\sum_{\vec{a}\in\{0,1\}^n, \vec{a}\ne \vec{0}, H(\vec{a})\le k} C_{\vec{a}} F_{\vec{a}}(\vec{v})} e^{\sum_{\vec{a}\in\{0,1\}^n, \vec{a}\ne \vec{0}, H(\vec{a})> k} C_{\vec{a}} F_{\vec{a}}(\vec{v})}}{\sum_{\vec{v}} e^{\sum_{\vec{a}\in\{0,1\}^n,\vec{a}\ne \vec{0}} C_{\vec{a}} F_{\vec{a}}(\vec{v})}} \nonumber \\
       &=& \underbrace{\left (\frac{\sum_{\vec{v}} e^{\sum_{\vec{a}\in\{0,1\}^n,\vec{a}\ne \vec{0},H(\vec{a})\le k} C_{\vec{a}} F_{\vec{a}}(\vec{v})}}{\sum_{\vec{v}} e^{\sum_{\vec{a}\in\{0,1\}^n,\vec{a}\ne \vec{0}} C_{\vec{a}} F_{\vec{a}}(\vec{v})}}\right)}_{\mathcal{N}\ne f(\vec{v}) \ge 0} \underbrace{\left (e^{\sum_{\vec{a}\in\{0,1\}^n, \vec{a}\ne \vec{0}, H(\vec{a})> k} C_{\vec{a}} F_{\vec{a}}(\vec{v})}\right)}_{\kappa(\vec{v})\ge 0}\underbrace{\left(\frac{e^{\sum_{\vec{a}\in\{0,1\}^n, \vec{a}\ne \vec{0}, H(\vec{a})\le k} C_{\vec{a}} F_{\vec{a}}(\vec{v})}}{\sum_{\vec{v}} e^{\sum_{\vec{a}\in\{0,1\}^n,\vec{a}\ne \vec{0},H(\vec{a})\le k} C_{\vec{a}} F_{\vec{a}}(\vec{v})}}\right)}_{\phi(\vec{v}):\mathbb{V}\to [0,1]}\nonumber \\
       &=& \mathcal{N} \kappa(\vec{v}) \frac{e^{\sum_{\vec{a}\in\{0,1\}^n, \vec{a}\ne \vec{0}, H(\vec{a})\le k} C_{\vec{a}} F_{\vec{a}}(\vec{v})}}{\sum_{\vec{v}} e^{\sum_{\vec{a}\in\{0,1\}^n,\vec{a}\ne \vec{0},H(\vec{a})\le k} C_{\vec{a}} F_{\vec{a}}(\vec{v})}} \nonumber \\
       &=& \mathcal{N}\kappa(\vec{v})\phi(\vec{v})\nonumber
    \end{eqnarray}
   \end{subproof}
   \item \textbf{Statement 2:} $\kappa(\vec{v})$ is a configuration dependant non-negative prefactor upper bounded as $|\log(\kappa(\vec{v}))| \le M\left(2^{n} -\sum_{j=0}^{k} \ ^nC_j \right)$
   \begin{subproof}[Subproof]
       It is clear from the functional form of $\kappa(\vec{v})$ in \textbf{Statement 1} above that $\kappa(\vec{v}) = e^{\sum_{\vec{a}\in\{0,1\}^n, \vec{a}\ne \vec{0}, H(\vec{a})> k} C_{\vec{a}} F_{\vec{a}}(\vec{v})} \ge 0$ and is also explicitly configuration $\vec{v}$ dependant. To establish the upper bound we can do the following
       \begin{eqnarray}
      \log(\kappa(\vec{v})) &=& \sum_{\vec{a}\in\{0,1\}^n, \vec{a}\ne \vec{0}, H(\vec{a})> k} C_{\vec{a}} F_{\vec{a}}(\vec{v}) \nonumber \\
      |\log(\kappa(\vec{v}))| &=& |\sum_{\vec{a}\in\{0,1\}^n, \vec{a}\ne \vec{0}, H(\vec{a})> k} C_{\vec{a}} F_{\vec{a}}(\vec{v})| \nonumber \\
       &\le& \sum_{\vec{a}\in\{0,1\}^n, \vec{a}\ne \vec{0}, H(\vec{a})> k}|C_{\vec{a}} F_{\vec{a}}(\vec{v})| \:\:\:\:\:\:\because \:Triangle \:\:Inequality \nonumber \\
       &\le& \sum_{\vec{a}\in\{0,1\}^n, \vec{a}\ne \vec{0}, H(\vec{a})> k}|C_{\vec{a}}| |F_{\vec{a}}(\vec{v})| \nonumber \\
       &\le& \sum_{\vec{a}\in\{0,1\}^n, \vec{a}\ne \vec{0}, H(\vec{a})> k}|C_{\vec{a}}| |\prod_i v_i^{a_i}| \:\:\:\:\:\: \vec{a} \in \{0,1\}^n \nonumber \\
       &\le& \sum_{\vec{a}\in\{0,1\}^n, \vec{a}\ne \vec{0}, H(\vec{a})> k}|C_{\vec{a}}| \:\:\:\:\:\:\:\:\:\:\:\:\because\: |\prod_i v_i^{a_i}|=1 \:\:as\:\:v_i \:\in \:\{1,-1\} \nonumber \\
       &\le& \sum_{\vec{a}\in\{0,1\}^n, \vec{a}\ne \vec{0}, H(\vec{a})> k}M \:\:\:\:\:\:\:\:\:\:\:\:\:\:\:\:M\myeq max(|C_{\vec{a}}|)_{\vec{a}\in\{0,1\}^n, \vec{a}\ne \vec{0}, H(\vec{a})> k} \ge 0 \nonumber \\
       &\le& M\left(\sum_{\vec{a}\in\{0,1\}^n, \vec{a}\ne \vec{0}} 1 - \sum_{\vec{a}\in\{0,1\}^n, \vec{a}\ne \vec{0}H(\vec{a})\le k} 1\right) \nonumber \\
       &\le& M\left(\sum_{\vec{a}\in\{0,1\}^n} 1 - \sum_{\vec{a}\in\{0,1\}^n,H(\vec{a})\le k} 1\right)\nonumber \\
       |\log(\kappa(\vec{v}))| &\le& M\left(2^{n} -\sum_{j=0}^{k} \ ^nC_j \right)
       \end{eqnarray}
       It must be noted that this upper bound is extremely generic and assumes no specific restrictions or structures on the set of coefficients $\{C_{\vec{a}}\}_{\vec{a}\in \{0,1\}^n}$. If however other features of the probability distribution is known such that if one can analytically establish $\{C_{\vec{a}}\}_{\vec{a}\in \{0,1\}^n, H(\vec{a})\le k}=f(k)$ i.e. may be exponentially decaying in the Hamming weight $k$, then one can leverage those to deduce stricter bounds. However such features are usually problem specific. Also the above upper bound when averaged over all configurations remains unchanged.
   \end{subproof}
\end{enumerate}
\end{proof}

\newpage
\section{Alteration in Variance of the Sampled estimate of the local energy due to parameterization of $\kappa(\vec{v})$}

We would like to prove in this section that certain choices/parameterization of $\kappa(\vec{v})$ introduced in previous section can lead to reduction in variance of the sampled estimate of a random variable. However, before we delve into that we first show
that the mean of the random variable using samples generated from the surrogate distribution of $\phi(\vec{v})$ is an unbiased estimator of the target sample estimate computed using the original distribution $P(\vec{v})=\rho^v_v(\vec{X})$ (see main text). In other words, mean of a random variable of interest($E_{loc}(v)$) over the distribution $\rho^v_v$ is equivalent to the mean of $E_{loc}(v)\kappa(v)$ over the distribution $\phi(v)$. We do this by using the following lemma

\begin{lemma}
\textbf{Means obtained from $\rho_v^v$ and $\phi(v)$}: The mean obtained by averaging $E_{loc}(v)$ over probability distribution $\rho_v^v$ is the same as the mean of  $E_{loc}(v)\kappa(v)$ over probability distribution $\phi(v)$.
\[
{\langle H \rangle = \langle E_{loc}(v) \rangle_{\rho_v^v} = \langle E_{loc}(v) \kappa(v) \rangle_{\phi(v)}}
\]
\end{lemma}

\begin{proof}
    
\begin{align}
    \langle H \rangle &= Tr(\rho H) \nonumber\\
    &= \sum_v \rho_v^v \Big( \frac{\sum_{v'} \rho_{v'}^v H_v^{v'}}{\rho_v^v} \Big) \nonumber\\
    &= \sum_v \rho_v^v E_{loc}(v) = \langle E_{loc}(v) \rangle_{\rho_v^v} \label{Eq. S1}
\end{align}
    

\begin{align}
&\langle E_{loc}(v) \kappa(v) \rangle_{\phi(v)} \nonumber\\
&= \sum_v E_{loc}(v) \kappa(v) \phi(v) \nonumber\\
&= \sum_v E_{loc}(v) \rho_v^v \nonumber\\ 
&= \langle E_{loc}(v) \rangle_{\rho_v^v} \label{Eq. S2} 
\end{align}


Thus from \ref{Eq. S1} and \ref{Eq. S2},
\begin{eqnarray}
\boxed{\langle H \rangle = \langle E_{loc}(v) \rangle_{\rho_v^v} = \langle E_{loc}(v) \kappa(v) \rangle_{\phi(v)}} \label{Eq. S3}
\end{eqnarray}

\end{proof}



Thus we see we have access to two distributions $\rho_v^v$ and $\phi(v)$. If we sample $v$ from either and compute the mean of two different random variables, i.e., $E_{loc}(v)$ for $\rho_v^v$ and $E_{loc}(v) \kappa(v)$ for $\phi(v)$, then the population mean of the two would be equal. The mean of either of these two random variables can act as a proxy of the other.\\

Now we would like to prove the main theorem in this Section

\begin{theorem}\label{thm:var_reduction}
\textbf{Variances obtained from $\rho_v^v$ and $\phi(v)$:}
For a specific form of $\kappa(v) = \frac{1}{\lambda | E_{loc}(v) |^\alpha}$, the difference between the variance of the sampled random variable $E_{loc}(v)$ over $\rho_v^v$ and the variance of $E_{loc}(v) \kappa(v)$ over $\phi(v)$ is given by, 
\[\sigma_{E_{loc}(v)}^2 - \sigma_{E_{loc}(v) \kappa(v)}^2 = \text{Cov} \left( | E_{loc}(v) |^{2-\alpha}, | E_{loc}(v) |^\alpha \right)\]
\end{theorem}

\begin{proof}
    
\begin{align*}
    \sigma_{E_{loc}(v)}^2 &= \sum_v E^2_{loc}(v) \rho_v^v - \langle E_{loc}(v) \rangle_{\rho_v^v}^2 \\
    \sigma_{E_{loc}(v) \kappa(v)}^2 &= \sum_v E_{loc}(v) \kappa(v)^2 \phi(v) - \langle E_{loc}(v) \kappa(v) \rangle_{\phi(v)}^2
\end{align*}

This can tell us how far are each individual estimate of random variable ($E_{loc}(v)$ for $\rho_v^v$ but $E_{loc}(v) \kappa(v)$ for $\phi(v)$) differ from their respective means (which by construction, see \ref{Eq. S3}, are equivalent). Also, in the respective Chebyshev inequality, it is the variance of the respective means that sits and it is not the same random variable. 

\[
P(|E_{loc}(v) \kappa(v) - \langle E_{loc}(v) \kappa(v) \rangle_{\phi(v)}| \geq \epsilon) \leq \frac{\text{Var}(E_{loc}(v) \kappa(v))}{N_s \, \epsilon^2}
\]

\[
P(|E_{loc}(v) - \langle E_{loc}(v) \rangle_{\rho_v^v}| \geq \epsilon) \leq \frac{\text{Var}(E_{loc}(v))}{N_s \, \epsilon^2}
\]

We need to see for the same $N_s$ if the variances $\text{Var}(E_{loc}(v) \kappa(v)), \text{Var}(E_{loc}(v))$ are equivalent or not?


\begin{align*}
    \text{Thus, } &\text{Var}(E_{loc}(v))_{\rho_v^v} - \text{Var}( E_{loc}(v) \kappa(v))_{\phi(v)} \\
    &= \sigma^2_{E_{loc}(v)} - \sigma_{E_{loc}(v) \kappa(v)}^2 \\
    &= \sum_v E^2_{loc}(v) \rho_v^v - \sum_v (E_{loc}(v) \kappa(v))^2 \phi(v) + \langle E_{loc}(v) \kappa(v) \rangle_{\phi(v)}^2 - \langle E_{loc}(v) \rangle_{\rho_v^v}^2 \\
    &= \sum_v E^2_{loc}(v) \rho_v^v - \sum_v (E_{loc}(v) \kappa(v))^2 \phi(v) \\ 
    &= \sum_v E^2_{loc}(v) \rho_v^v - \sum_v E^2_{loc}(v) \kappa^2(v) \phi(v) \\
    &= \sum_v E^2_{loc}(v) \rho_v^v - \sum_v \frac{E^2_{loc}(v) \left( \rho_v^v \right)^2}{\phi(v)} \\
    &= \sum_v E^2_{loc}(v) \rho_v^v \left( 1 - \frac{\rho_v^v}{\phi(v)} \right)
\end{align*}








Now if we assume \quad $\boxed{\kappa(v) = \frac{1}{\lambda | E_{loc}(v) |^\alpha} \implies \phi(v)= \lambda | E_{loc}(v) |^\alpha \rho_v^v}$

\begin{align*}
\text{then} & \quad \lambda = \frac{\sum_v \phi(v)}{\sum_v | E_{loc}(v) |^\alpha \rho_v^v} = \frac{1}{\sum_v | E_{loc}(v) |^\alpha \rho_v^v} \\
\text{or} & \quad \phi(v) = \frac{| E_{loc}(v) |^\alpha \rho_v^v}{\sum_{v'} | E_{loc}(v') |^\alpha \rho_{v'}^{v'}} \\
\text{or} & \quad \frac{\rho_v^v}{\phi(v)} = \frac{\sum_{v'} | E_{loc}(v') |^\alpha \rho_{v'}^{v'}}{| E_{loc}(v) |^\alpha}
\end{align*}


In other words, we have

\begin{align}
& \sigma_{E_{loc}(v)}^2 - \sigma_{E_{loc}(v) \kappa(v)}^2 \nonumber\\ 
&= \sum_v E^2_{loc}(v) \rho_v^v \left( 1 - \frac{\sum_{v'} | E_{loc}(v') |^\alpha \rho_{v'}^{v'}}{| E_{loc}(v) |^\alpha} \right) \nonumber\\
&= \sum_v E^2_{loc}(v) \rho_v^v \left( 1 - \frac{\langle | E_{loc}(v) |^\alpha \rangle_{\rho_{v}^v}}{| E_{loc}(v) |^\alpha} \right) \\
&= \sum_v E^2_{loc}(v) \rho_v^v - \left( \sum_{v} \frac{E^2_{loc}(v) \rho_{v}^v}{| E_{loc}(v) |^\alpha} \right) \left\langle | E_{loc}(v) |^\alpha \right\rangle_{\rho_v^v} \\
&= \left\langle | E_{loc}(v) |^2 \right\rangle_{\rho_v^v} - \left\langle | E_{loc}(v) |^{2-\alpha} \right\rangle_{\rho_v^v} \left\langle | E_{loc}(v) |^\alpha \right\rangle_{\rho_v^v} \\
&= \text{Cov} \left( | E_{loc}(v) |^{2-\alpha}, | E_{loc}(v) |^\alpha \right)
\end{align}

Or
\begin{equation}
\boxed{\sigma_{E_{loc}(v)}^2 - \sigma_{E_{loc}(v) \kappa(v)}^2 = \text{Cov} \left( | E_{loc}(v) |^{2-\alpha}, | E_{loc}(v) |^\alpha \right)}  \label{Eq. S6}
\end{equation}

\end{proof}


\begin{lemma}
Let’s define
\[
Z  = | E_{loc}(v) |^\alpha \qquad and \qquad
g(Z) = | E_{loc}(v) |^{2-\alpha} = Z^{2/\alpha - 1} \nonumber 
\]


Such that from Eq.\ref{Eq. S6}, 
\[{\sigma_{E_{loc}(v)}^2 - \sigma_{E_{loc}(v) \kappa(v)}^2 = \text{Cov} \left( Z^{2/\alpha - 1}, Z \right )
= \text{Cov} \left( g(Z), Z \right )}
\]

Given \(Z \geq 0\), then show that \(\text{Cov}(g(Z), Z) > 0\) if \(g(Z)\) is non-decreasing, i.e., \(g'(Z) \geq 0\).
\end{lemma}

\begin{proof}
    
\begin{align*}
\text{Cov}(g(Z), Z) &= \mathbb{E}[(g(Z) - \langle g(Z) \rangle)(Z - \langle Z \rangle)] \\
&\text{This expectation and } \langle \rangle \text{ is over } \rho_v^v \text{ as is the covariance.} \\
&= \mathbb{E}\left[(Z - \langle Z \rangle)(g(Z) + g(\langle Z \rangle) - g(Z) - \langle g(Z) \rangle)\right] \\
&= \mathbb{E}[(Z - \langle Z \rangle)(g(Z) - g(\langle Z \rangle))]
+ \mathbb{E}[(Z - \langle Z \rangle)(g(\langle Z \rangle) - \langle g(Z) \rangle)] \\
&= \mathbb{E}[(Z - \langle Z \rangle)(g(Z) - g(\langle Z \rangle))] + \mathbb{E}[(Z - \langle Z \rangle)(g(\langle Z \rangle) - \langle g(Z) \rangle)]
\end{align*}



$g(\langle Z \rangle) - \langle g(Z) \rangle$ is independent of \(Z\) and hence can be taken outside expectation. This is because we have \(g\) at \(\langle Z \rangle\) and \(\langle g(Z) \rangle\), both independent of \(Z\) now and just scalars.


\begin{align*}
\Rightarrow \mathbb{E}[(Z - \langle Z \rangle)&(g(Z) - g(\langle Z \rangle))] + (g(\langle Z \rangle) - \langle g(Z) \rangle) \mathbb{E}[(Z - \langle Z \rangle)] \\
&\text{As} \quad \mathbb{E}[(Z - \langle Z \rangle)] = 0 \\
\Rightarrow \mathbb{E}[(Z - \langle Z \rangle)&(g(Z) - g(\langle Z \rangle))]
\end{align*}

Now as \(g(Z)\) is a non-decreasing function,

\begin{align*}
&\frac{g(Z_1) - g(Z_2)}{Z_1 - Z_2} \geq 0 \qquad \text{(follows from derivative itself)} \\
\Rightarrow &(Z_1 - Z_2)^2 \frac{g(Z_1) - g(Z_2)}{Z_1 - Z_2} \geq 0 \quad \\
\Rightarrow &(g(Z_1) - g(Z_2))(Z_1 - Z_2) \geq 0 \\
\end{align*}



Thus,
\[
\mathbb{E}[(g(Z) - g(\langle Z \rangle))(Z - \langle Z \rangle)] \geq 0
\]
where \(Z_1 \to Z\), \(Z_2 \to \langle Z \rangle\). \newline

Now if we take \(Z_1 = Z\) and \(Z_2 = \langle Z \rangle\), then that means
\[
\mathbb{E}[(g(Z) - g(\langle Z \rangle))(Z - \langle Z \rangle)] \geq 0
\]

So,
\[
\boxed{\text{Cov}(g(Z), Z) \geq 0}
\]

\end{proof}


\begin{lemma}
In the range \(0 \leq \alpha \leq 2\), show that \[\sigma_{E_{loc}(v)}^2 - \sigma_{E_{loc}(v) \kappa(v)}^2 = \text{Cov}(Z^{2/\alpha - 1}, Z) \geq 0 \]
\end{lemma}

\begin{proof}

When \(2 < \alpha\),
\(Z^{2/\alpha - 1}\) is a non-increasing function as
\[
g'(Z) = \frac{(2/\alpha - 1)}{Z} Z^{2/\alpha - 1} \leq 0 \quad [Z \geq 0]
\]

If that is the case, by the exact same logic in the last statement, we can say
\[
\text{Cov}(g(Z), Z) \leq 0
\]


Also, \(\alpha \leq 0\)

Even then \(g(Z) = Z^{2/\alpha - 1}\) is non-increasing as
\[
g'(Z) = \frac{1}{Z} Z^{2/\alpha - 1} \geq 0 \quad \text{but non-increasing}
\]

Thus even then
\[
\text{Cov}(g(Z), Z) \leq 0
\]


So we have a full characterization:

\[
\sigma_{E_{loc}(v)}^2 - \sigma_{E_{loc}(v) \kappa(v)}^2 = \text{Cov}(|E_{loc}(v)|^{2-\alpha}, |E_{loc}(v)|^\alpha)
= \text{Cov}(Z^{2/\alpha - 1}, Z), \quad (Z = |E_{loc}(v)|^\alpha \geq 0)
\]

\[
\begin{array}{|c|c|c|c|}
\hline
\alpha < 0 & 0 \leq \alpha \leq 2 & \alpha > 2 \\
\hline
Z^{2/\alpha - 1} \text{ is non-increasing} & Z^{2/\alpha - 1} \text{ is non-decreasing} & Z^{2/\alpha - 1} \text{ is non-increasing} \\
\text{Cov}(Z^{2/\alpha - 1}, Z) \leq 0 & \text{Cov}(Z^{2/\alpha - 1}, Z) \geq 0 & \text{Cov}(Z^{2/\alpha - 1}, Z) \leq 0 \\
\hline
\end{array}
\]

Special points:

\[
\begin{array}{|c|c|c|c|}
\hline
\alpha = 0 & \alpha = 1 & \alpha = 2 \\
\hline
\text{Cov}(Z^{2/\alpha - 1}, Z) = 0 & \text{Cov}(Z, Z) \geq 0 & \text{Cov}(1, Z) = 0 \\
\hline
\end{array}
\]

\end{proof}

\newpage
\section{Data driven construction of the surrogate probability distribution }
In the previous sections, we have shown that any probability distribution can be written in terms of an exponential of an arbitrary-degree polynomial function of the spin configuration (see Theorem \ref{thm:factor_theo_main} in main manuscript or Theorem \ref{thm:factor_theo} in Section S2). So far we have proven that it is always possible to represent an arbitrary probability distribution $P(\vec{v}):\mathbb{V}\to[0,1]$ as $\mathcal{N}\kappa(\vec{v})\phi(\vec{v})$ where $\phi(\vec{v})$ is the distribution of the surrogate network defined as 

\[\phi({\vec{v}}, \vec{X}) \propto exp \left(-\beta \sum_i l_i(\vec{X}) v_i + \sum_{ij} J_{ij}(\vec{X}) v_i v_j + \cdots \right) \]
Here $l_i \in \vec{l}(\vec{X})$ denotes the on-site field at each spin of the surrogate network and  $J_{ij}(\vec{X}) \in \vec{J}(\vec{X})$ denotes the mutual coupling between a pair of spins of the surrogate. The kernel function \( \kappa(\vec{v}) \) accounts for additional flexibility in the construction when the surrogate distribution is truncated to a finite degree $k$. 
In this section we provide a concrete data-driven recipe to construct the same. Here we choose $P(\vec{v})=\rho_v^v$ where $\rho_v^v$ is the probability distribution associated with graph $G_1$ in main text. The surrogate so defined (see graph $G_2$ in main text) has been truncated to second-degree i.e. $k=2$. Even though the algorithm we describe below works for any value of $k$, however, it should be noted that implementing a polynomial of degree \( k \) requires \( k \)-qubit gates in a quantum circuit, and hence a large $k$ can increase the complexity of the circuit. Therefore, one would prefer to work with low-degree polynomials for practical implementation which motivates the choice for $k=2$. At the end of this protocol we shall have a decomposition of the following kind

\[
 \rho_v^v(v,\vec{X},\beta) = \phi(v,\vec{X}) \kappa(v,\vec{X},\beta) \; + \epsilon(v) 
 \]
 where $\epsilon(\vec{v})$ is the fitting error associated with the approximation $\rho_v^v \approx \phi(v) \kappa(v)$.


\subsection{Sampling configurations for fitting}
To ensure that the surrogate distribution closely approximates the actual distribution, we aim to fit the parameters of the polynomial function through a non-linear fitting process. Ideally, we would use all possible configurations to fit these parameters. However, this approach would be computationally prohibitive and moot the advantage offered by our algorithm. Therefore, we strategically select a sample of configurations to perform the fitting. The selected configurations include:
\begin{itemize}
    \item \textbf{The configurations with largest $\rho_v^v$ value}: These configurations are important because they are the most probable and thus contribute significantly to the overall distribution. We can find these configurations efficiently using the algorithm outlined below. Section \ref{sec: draw samples} Selecting these ensures that the model accurately represents the high-probability regions of the distribution.
    
    \item \textbf{Random configurations}, which represent the bulk of the less probable configurations. Incorporating these ensures the model generalizes well and prevents overfitting to the most probable configurations alone.
\end{itemize}

Focusing on the most probable configurations is crucial because they dominate the sampling process; these configurations have the highest \( \rho_v^v \) values and will appear most frequently during sampling. When fitting the surrogate distribution \( \phi(v) \), we need to ensure it performs well for these configurations since they represent the most critical part of the probability landscape. However, there is a risk associated with only using the best configurations: overfitting. If the model fits exclusively to high-probability configurations, it may perform poorly on low-probability ones. In extreme cases, this could cause the model to predict low-probability configurations so inaccurately that they interfere with high-probability predictions, compromising the overall accuracy. [See Fig \ref{fig:sample_choice}] This is why we include random configurations in the sample. These configurations, though less probable, help to ensure that the surrogate distribution generalizes well across the entire probability space. They prevent the model from overfitting to the high-probability regions and ensure that \( \phi(v) \) adequately represents both high- and low-probability configurations. By constructing the surrogate network with this balanced set of configurations—both the most probable and a representative sample of lower-probability ones—we create a robust model that maintains accuracy across the full distribution without being biased toward a narrow set of configurations. We keep the split between the best configuration and random configuration as a hyperparameter. In this work, we maintain a ratio of 25\% best configurations and 75\% random configurations. The total number of samples is kept below \( O(n^2) \).\\

\subsection{Algorithm for finding q-largest configurations} \label{sec: draw samples}
Given a probability distribution $\rho_v^v$ which has its support over $O(2^n)$ spin configurations as $\vec{v} \in \mathbb{V}$. We need to find $q$-spin configurations that maximize $\rho_v^v$ which from the main text is defined as:  
$\rho_v^v = \text{exp} \Big(-2\beta \sum_{i=1}^n Re(a_i) v_i \Big) \prod_{j=1}^m |\cosh(b_j + \sum_{i=1}^n w_{ij} v_i)|^2$
We exploit the structure of the probability distribution to find the $q$-largest configurations.
Note that the first exponential term is maximized for a configuration $\{v_i\}^{(0)} = \{-\text{sgn}(a_i)\}$ for \( i = 1, \dots, n \). Individual cosh terms are maximized when $\{v_i\}^{(j)} = \text{sgn}(\text{Re}(w_{ij}))$ or $\{v_i\}^{(m+j)} = -\text{sgn}(\text{Re}(w_{ij}))$ for \( j = 1, \dots, m \). Starting from an initial set of 2m+1 candidates, we run an iterative algorithm that perturbs these configurations to identify the approximate $q$-largest configurations.\\ 

\textbf{Algorithm Steps}
\begin{enumerate}
    \item Initialize a list \textbf{config} to store the candidate spin configurations and another list \textbf{result} to store the spin configuration with the largest $\rho_v^v$. 
    \item Choose \{-sgn(Re(a)), sgn(Re($w_j$)), and -sgn(Re($w_j$))\} as the first (2m+1) contenders to the largest spin configurations and store them in \textbf{config}. These seeds are chosen because they individually maximize different terms in $\rho_v^v$. 
    \item Compute $\rho_v^v$ (upto normalization) for each (2m+1) configuration, sort them in descending order of their $\rho_v^v$, and remove any duplicates.  
    \item Run q iterations: 
    \begin{enumerate}
        \item Select a configuration $v_i$ $\in$ \textbf{config} with the largest $\rho_v^v$ and store it in \textbf{result}. Remove it from \textbf{config}. 
        \item Perform single-site perturbations on $v_i$, flipping each spin to generate $n$ new candidate configurations $\{v_i^j\}_{j=1}^n$. 
        \item For each new configuration $v_i^j$ not already in config, compute $\rho_v^v(v_i^j)$ and merge it into \textbf{config} based on its $\rho_v^v(v_i^j)$ values. 
        \item Trim \textbf{config} to maintain only the top q candidates by discarding configurations with small $\rho_v^v(v_i^j)$ values.
    \end{enumerate}
    \item Return \textbf{result} - the approximate q-best configurations.\\
\end{enumerate}

\textbf{Time Complexity}: 
The algorithm runs in \( O(q n)\). At each iteration, we generate \( O(n) \) new configurations and insert each into a sorted list. Depending on the choice of data structure, the insertion operation can take \( O(\log(q)) \) time for an array or can be done in $O(1)$ time by having a heap implementation. The nature of the algorithm is heuristic and works well with high probability for general parameter sets we dealt with so far. However, it might be possible to curate special instances where it might fail to obtain the best k configurations .\\

\begin{figure}[ht]
    \centering
        \includegraphics[width=1.0\textwidth]{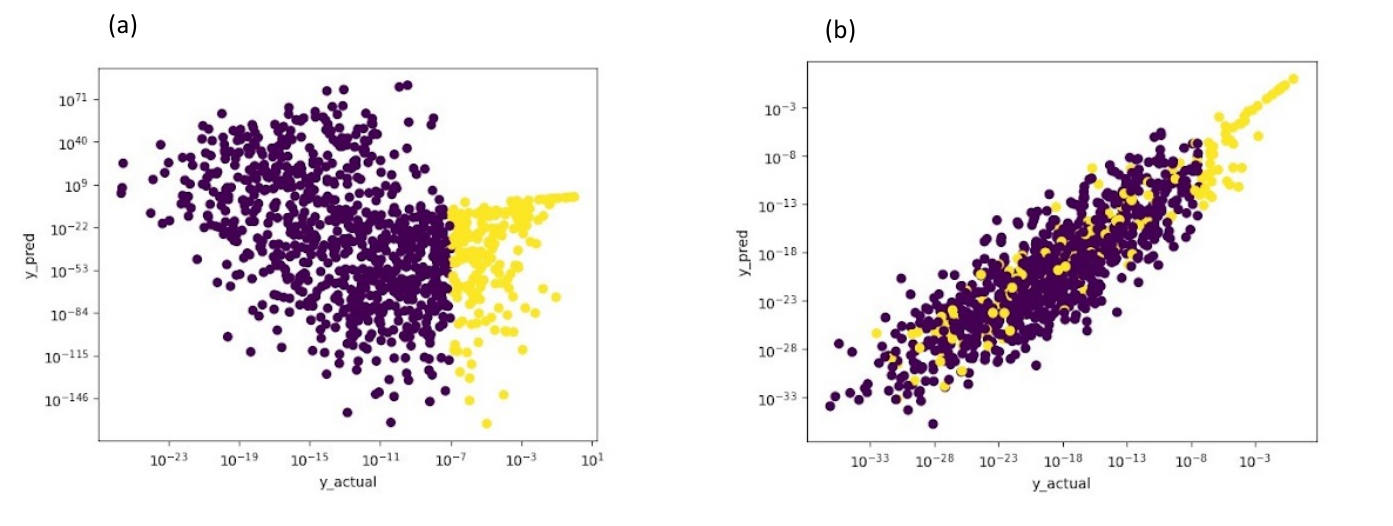}
    \caption{The figure demonstrates how the choice of configuration significantly impacts the quality of the model's predictions. In both plots, yellow points represent configurations used during fitting, while blue points denote configurations not included in the fitting process and are unseen to the model. An ideal fitting would be seen as a straight line with the prediction values ($y_{\text{pred}}$) exactly equal to the actual values ($y_{\text{actual}}$). (a) The left figure shows result of fitting performed using only the top configurations. While the model predicts fairly well for the selected configurations (yellow points), it fails dramatically for unseen configurations (blue points), especially the prediction for lower value points exceed even the high value points in some cases. (b) In the right Figure, we incorporate a mix of random configurations alongside the top configurations, and it can be seen that the model maintains accurate predictions for high values (large $y_{\text{actual}}$) while keeping small values in check. Also note how the algorithm for selecting the best configurations in Section \ref{sec: draw samples} ensures that all configurations with large $y_{\text{actual}}$ values are chosen accurately, as evidenced by the consistent presence of yellow points on the far-right side of the plots.}
    \label{fig:sample_choice}
\end{figure}

\subsection{Final Fitting Algorithm}\label{Sec: Fit Algo}


We achieve the best fit using a two-fold optimization protocol. First, we fit the logarithm of the $\rho_v^v$ distribution to a polynomial model through a weighted least-squares (LSQ) fitting procedure. Then, we refine the results using non-linear optimization techniques, such as the Broyden–Fletcher–Goldfarb–Shanno (BFGS) algorithm. This two-step approach is necessary because the optimization landscape is highly non-convex, making the results of non-linear optimization highly sensitive to the initial parameter values. The LSQ fitting step provides a robust starting point by converging to an optimal solution with high reliability. This initial solution is then used as input for the non-linear optimization step, which fine-tunes the parameters to achieve a more accurate fit. This meticulous process is crucial because the fitting procedure plays an essential role in constructing the surrogate network. Ensuring that this mapping is as precise as possible is fundamental to the success of the entire framework.\\

\textbf{Polynomial fitting:}
We transform the above non-linear fitting problem into a polynomial fitting problem by taking a logarithm of both sides. 
\[\text{log} \rho_v^v(v,\vec{X}) \approx c_0 + \sum_i l_i(\vec{X}) v_i + \sum_{i,j} J_{ij}(\vec{X}) v_i v_j + \cdots \]

The polynomial to be obtained by fitting has a constant term \(c_0\), a linear term involving \(v_i\), with coefficients \(l_i(\vec{X})\),
and quadratic terms involving \(v_i v_j\), with coefficients \(J_{ij}(\vec{X})\). Note that we restrict our discussion to polynomials with up to quadratic terms, but the method below is general for any order of polynomials. 

We aim to find the optimal values for \(c_0\), \(l_i\), and \(J_{ij}\) from samples ($\{v^{i}\}$) drawn from the distribution \(\rho_v^v(v)\) as mentioned in Section \ref{sec: draw samples}. However, we want to avoid having small values of \(\rho_v^v(v)\) (which leads to highly negative log values) dominate the fitting process. Therefore, we perform a weighted fit, where the weights are proportional to \(\rho_v^v(v)\). This gives more influence to higher values of \(\rho_v^v(v)\) during the fitting process. 
The fitting process solves a weighted least squares (lsq) problem to find the coefficients \(c_0\), \(l_i\), and \(J_{ij}\) that best approximate the function. Mathematically, we aim to minimize the weighted sum of squared residuals:
\[
\min_{\theta} \left\| T (A \theta - Y) \right\|_2
\]

Where:
\begin{itemize}
    \item \(\theta = [c_0, l_1, l_2, \dots, l_N, J_{11}, J_{12}, \dots, J_{NN}]\) represents the vector of coefficients (parameters).
    \item \(Y\) is the vector of actual values \(\{\log \rho_v^v(v^{(i)})\}\) for the sample set \(\{v^{(i)}\}\).
    \item \(T\) is a diagonal weight matrix, where \(T_{ii} = \rho_v^v(v^{(i)})\).
\end{itemize}

The matrix \(A\) is the design matrix that contains the constant, linear, and quadratic terms. Its structure is:
\[
A = \begin{bmatrix}
1 & v_1^{(1)} & v_2^{(1)} & \dots & v_N^{(1)} & (v_1^{(1)})^2 & v_1^{(1)} v_2^{(1)} & \dots & (v_N^{(1)})^2 \\
1 & v_1^{(2)} & v_2^{(2)} & \dots & v_N^{(2)} & (v_1^{(2)})^2 & v_1^{(2)} v_2^{(2)} & \dots & (v_N^{(2)})^2 \\
\vdots & \vdots & \vdots & \vdots & \vdots & \vdots & \vdots & \vdots \\
1 & v_1^{(k)} & v_2^{(k)} & \dots & v_N^{(k)} & (v_1^{(k)})^2 & v_1^{(k)} v_2^{(k)} & \dots & (v_N^{(k)})^2 \\
\end{bmatrix}
\]


The solution to this problem is given by:
\[
\theta = (A^T T A)^{-1} A^T T Y
\]

\textbf{Non-linear fitting:}
The optimal values for the parameters \(\theta = \{c_0, l_i, J_{ij}\} \), obtained from the preceding least-squares optimization, serve as a starting point for the non-linear optimization procedure. This step refines the parameter estimates by directly minimizing the difference between the distributions, yielding more accurate results, in contrast to the previous step, where the logarithms of the distributions were compared.

The cost function for this optimization is defined as:  
\[ 
C = \min_\theta \left\| \exp(P_k(\theta,v)) - \rho_v^v(v) \right\|_2 
\]  
where \(\exp(P_k(\theta,v))\) represents the approximate probability distribution represented by the exponential of the Polynomial dependent on the parameters $\theta$, and \(\rho_v^v(v)\) is the target distribution. The goal is to minimize the \(L_2\)-norm of the difference between these distributions. 
In our work, we employ the BFGS algorithm, a quasi-Newton optimization method, to minimize the cost function. BFGS uses both gradient information and an approximation of the Hessian matrix to guide the search for the optimal parameters. 
Gradients are approximated using finite differences, allowing BFGS to iteratively refine the parameters despite the absence of explicit derivatives. The result of the non-linear fitting provides the optimal values for the parameters \(c_0\), \(l_i\), and \(J_{ij}\). These parameters can then be used to construct the values of \(\phi_v(\vec{X})\) for any given configuration \(v\).

\newpage
\section{Gradient Expressions}

In this section, we derive the gradient expressions for the state ansatz with respect to its parameters. These gradients are crucial for training the model parameters. We provide analytical forms for the gradient of the density matrix and the expectation value of observables. 

\begin{lemma}  
$\partial_{x_i} \rho_{v'}^v = \mathcal{D}_{v'}^v (x_i) \odot \rho_{v'}^v$; where the '$\odot$' represents element-wise Hadamard product between matrices and the matrix $\mathcal{D}_{v'}^v (x_i)$ for various parameters $x_i$ is given as follows:
\end{lemma}

\[
\begin{array}{|c|c|}
\hline
x_i & \mathcal{D}^v_{v_i'} (x_i) \\ 
\hline
\text{Re}(a_k) &  -\beta \left( v_k + v_k' \right) \\ 
\hline
\text{Im}(a_k) &  -i\beta \left( v_k - v_k' \right) \\ 
\hline
\text{Re}(b_p) &  \beta \Big\{ \tanh\left( \beta b_p + \beta \sum_{i=1}^n W_{ip} v_i \right) \\ 
& + \tanh\left( \beta b_p^* + \beta \sum_{i=1}^n W_{ip}^* v_i' \right) \Big\} \\ 
\hline
\text{Im}(b_p) &  i\beta \Big\{ \tanh\left( \beta b_p + \beta \sum_{i=1}^n W_{ip} v_i \right) \\ 
& - \tanh\left( \beta b_p^* + \beta \sum_{i=1}^n W_{ip}^* v_i' \right) \Big\} \\ 
\hline
\text{Re}(W_{kp}) & \beta \Big\{ \tanh\left( \beta b_p + \beta \sum_{i=1}^n W_{ij} v_i \right) v_k  \\ 
& + \tanh\left( \beta b_p^* + \beta \sum_{i=1}^n W_{ip}^* v_i' \right) v_k'  \Big\} \\ 
\hline
\text{Im}(W_{kp}) & i \beta \Big\{ \tanh\left( \beta b_p + \beta \sum_{i=1}^n W_{ip} v_i \right) v_k  \\ 
& - \tanh\left( \beta b_p^* + \beta \sum_{i=1}^n W_{ip}^* v_i' \right) v_k'  \Big\} \\ 
\hline


\end{array}
\]

\begin{proof}
The goal is to compute the derivative of \(\rho^v_{v'}\) with respect to various parameters \(x_i\). The state \( \rho(\vec{x}) \) is defined as:

\begin{align*}
\rho^v_{v^\prime}(\vec{x}) \propto exp \left({-\beta \sum_i a_i v_i + \sum_i a_i^* v_i^\prime}\right)  \:\prod_{j=1}^m \Gamma_j (\vec{v},\vec{b}, \vec{W}) \:\Gamma_j (\vec{v}^\prime,\vec{b^*}, \vec{W^*}) \label{eq:Ising_state}
\end{align*} 

where, \(\Gamma_j(\vec{v}, \vec{b}, \vec{W}) = \cosh \left[ \beta \left(b_j + \sum_{i=1}^n W_{ij} v_i\right) \right]\)

We take the logarithmic derivative of \(\rho^v_{v'}\) with respect to \(x_i\):

\[
\boxed{\partial_{x_i} \rho^v_{v'} = \rho^v_{v'} \cdot \partial_{x_i} \ln \rho^v_{v'}}
\]

We denote $\mathcal{D}_{v'}^v (x_i) = \partial_{x_i} \ln \rho^v_{v'} $. 
Therefore $\partial_{x_i} \rho_{v'}^v = \mathcal{D}_{v'}^v (x_i) \odot \rho_{v'}^v$; where the '$\odot$' represents element-wise Hadamard product between matrices and the matrix $\mathcal{D}_{v'}^v (x_i)$ 

The logarithm of \(\rho^v_{v'}\) is:

\[
\ln \rho^v_{v'} = -\beta \sum_i \left( a_i v_i + a_i^* v_i^\prime \right) + \sum_{j=1}^m \left[ \ln 
\Gamma_j (\vec{v},\vec{b}, \vec{W})
\: +\ln \Gamma_j (\vec{v}^\prime,\vec{b^*}, \vec{W^*})
\right]
\]

Now, we take the derivative of each term w.r.t the parameters $x_i$: 

- For the first term:

\[
\partial_{x_i} \left[ -\beta \sum_i \left( a_i v_i + a_i^* v_i^\prime \right) \right]
\]

is nonzero only if \(x_i\) corresponds to \(\mathrm{Re}(a_k)\) or \(\mathrm{Im}(a_k)\), giving the contributions:

\[
\boxed{\partial_{\mathrm{Re}(a_k)} = -\beta (v_k + v_k^\prime)} \qquad \boxed{\partial_{\mathrm{Im}(a_k)} = -i\beta (v_k - v_k^\prime)}
\]

- For the second term, 
\[
\ln \Gamma_j = \ln \cosh\left(\beta b_j + \beta \sum_{i=1}^n W_{ij} v_i\right),
\]

depends only $b$ and $W$ parameters. 
\begin{itemize}
\item For \(b_p\):

Using \(\tanh x = \partial_x \left(\ln \cosh x\right)\), we have:
\[
\partial_{\mathrm{Re}(b_p)} \ln \Gamma_j = \tanh\left(\beta b_p + \beta \sum_{i=1}^n W_{ip} v_i \right).
\]
Similarly, for the conjugate term from \(\Gamma_j(\vec{v}', \vec{b}^*, \vec{W}^*)\):
\[\partial_{\mathrm{Re}(b_p)} \ln \Gamma_j^\prime = \tanh\left(\beta b_p^* + \beta \sum_{i=1}^n W_{ip}^* v_i^\prime \right).
\]
Adding these:
\[
\boxed{{\partial_{\mathrm{Re}(b_p)} \ln \rho^v_{v'} = \beta \left[\tanh\left(\beta b_p + \beta \sum_{i=1}^n W_{ip} v_i \right) + \tanh\left(\beta b_p^* + \beta \sum_{i=1}^n W_{ip}^* v_i^\prime \right)\right]}}.
\]
Similarly for \(\mathrm{Im}(b_p)\),
\[
\boxed{{\partial_{\mathrm{Im}(b_p)} \ln \rho^v_{v'} = i \beta \left[\tanh\left(\beta b_p + \beta \sum_{i=1}^n W_{ip} v_i \right) - \tanh\left(\beta b_p^* + \beta \sum_{i=1}^n W_{ip}^* v_i^\prime \right)\right]}}.
\]

\item For \(W_{kp}\):

The derivative includes the dependence on \(v_k\) and \(v_k'\):
\[
\boxed{\partial_{\mathrm{Re}(W_{kp})} \ln \rho^v_{v'} = \beta \left[\tanh\left(\beta b_p + \beta \sum_{i=1}^n W_{ip} v_i \right)v_k + \tanh\left(\beta b_p^* + \beta \sum_{i=1}^n W_{ip}^* v_i^\prime \right)v_k^\prime\right]},
\]
and similarly for \(\mathrm{Im}(W_{kp})\):
\[
\boxed{\partial_{\mathrm{Im}(W_{kp})} \ln \rho^v_{v'} = i\beta \left[\tanh\left(\beta b_p + \beta \sum_{i=1}^n W_{ip} v_i \right)v_k - \tanh\left(\beta b_p^* + \beta \sum_{i=1}^n W_{ip}^* v_i^\prime \right)v_k^\prime\right]}.
\]

\end{itemize}

\end{proof}

\newpage
\begin{lemma}
    $\partial_{x_i} \langle O \rangle (\vec{X}) = \left\langle \left[\mathcal{D}_{x_i} \odot O^T\right]_\rho (v) \right\rangle_{\rho_v^v} - \left\langle \mathcal{D}_{x_i}(v) \right\rangle_{\rho_v^v} \left\langle O_{\rho} (v) \right\rangle_{\rho_v^v}$
\end{lemma}

\begin{proof}

\begin{align*}
\langle O \rangle (\vec{X}) &= \frac{Tr (O \rho (\vec{X}))}{Tr (\rho (\vec{X}))} \\
\partial_{x_i} \langle O \rangle (\vec{X}) &= \partial_{x_i} \left(  \frac{Tr (O \rho (\vec{X}))}{Tr (\rho (\vec{X}))} \right) \\
&= \partial_{x_i} \left( \frac{\sum_{v, v'} \rho_{v'}^v O_v^{v'}}{\sum_v \rho_v^v (\vec{X})} \right) \\
&= \frac{ \left( \sum_{v, v'} \partial_{x_i} \rho_{v'}^v O_v^{v'} \right)}{\sum_v \rho_v^v (\vec{X})}  -   \frac{\left( \sum_{v} \partial_{x_i} \rho_{v}^v \right) \left( \sum_{v,v'} \rho_{v'}^{v} O_v^{v'} \right)}{\left( \sum_v \rho_v^v (\vec{X}) \right)^2}  \\
&= \frac{\sum_{v, v'} \mathcal{D}_{v'}^v (x_i) \rho_{v'}^v O_v^{v'} }{\sum_v \rho_v^v (\vec{X})}  -  \left( \frac{ \sum_{v} \mathcal{D}_{v}^v (x_i) \rho_{v}^v }{\sum_v \rho_v^v (\vec{X})} \right) \left( \frac{ \sum_{v,v'} \rho_{v'}^{v} O_v^{v'} }{\sum_v \rho_v^v (\vec{X}) } \right)  \\
&= \sum_v \rho_v^v \left( \frac{\sum_{v'} \rho_{v'}^v \left[\mathcal{D}_{x_i} \odot O^T\right]_{v}^{v'} }{\rho_v^v} \right) / \sum_v \rho_v^v (\vec{X})  - \left\langle \mathcal{D}_{x_i}(v) \right\rangle_{\rho_v^v} \left\langle O_{\rho} (v) \right\rangle_{\rho_v^v}  \\
&= \left\langle \left[\mathcal{D}_{x_i} \odot O^T\right]_\rho (v) \right\rangle_{\rho_v^v} - \left\langle \mathcal{D}_{x_i}(v) \right\rangle_{\rho_v^v} \left\langle O_{\rho} (v) \right\rangle_{\rho_v^v}  \\
\end{align*}
By substituting $O=H$, i.e. the hamiltonian of the driver system in the above expression, we get analytical gradients used for training parameters $\vec{X}$.
\end{proof}

\newpage
\section{Quantum Circuits}

The sampling circuit for quantum proposals requires a \( k \)-qubit generalization of the \( R_{zz} \) gate, denoted as \( R_{zzz\ldots z} \). Here, \( k \) represents the order of interactions allowed in the surrogate network, with \( R_{zzz\ldots z} \) acting on \( k \)-qubits. This section outlines two methods to implement these gates using different kinds of two-qubit entangling gates along with arbitrary single-qubit unitary gates. 

\textbf{CNOT gate + arbitrary single qubit unitary operations }: The first approach utilizes a gate set composed of the CNOT gate and arbitrary single-qubit rotations. The decomposition of the \( R_{zzz\ldots z} \) gate into this gate set is as follows:

$$R_{zzz..z}(\theta) = \prod_{i=0}^{k-2} CNOT(i,i+1) \;\;  R_z^k(\theta) \;\; \prod_{i=k-2}^{0} CNOT(i,i+1)$$

where the CNOT gates create the necessary entanglement structure, and the \( R_z^k(\theta) \) gate applies the z-rotation on the \( k \)-th qubit. This decomposition is depicted in Figure \ref{fig:RZZZ_CNOT}.

\begin{figure}[h!]
    \centering
    \includegraphics[width=0.40\linewidth]{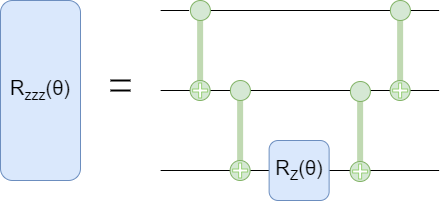}
    \caption{Decomposition of \( R_{zzz\ldots z} \) gate into CNOT gates and arbitrary single-qubit rotations.}
    \label{fig:RZZZ_CNOT}
\end{figure}

\textbf{Echoed Cross-Resonance (ECR) gate + arbitrary single qubit unitary operations}: 
An alternative implementation of the \( R_{zzz\ldots z} \) gate utilizes the Echoed Cross-Resonance (ECR) gate as the two-qubit entangling operation, combined with arbitrary single-qubit rotations. The ECR gate is functionally equivalent to the CNOT gate, differing only by additional single-qubit rotations. However, they are typically preferred over CNOT gates as they are native interactions in certain quantum architectures and often exhibit lower error rates.

A known relationship between the CNOT and ECR gates allows us to express the \( R_{zzz\ldots z} \) gate in terms of ECR gates\cite{tan2024quantum}. Specifically, the relation is:

\[
\text{CNOT} = \left[ R_Z(-\pi/2) \otimes R_Z(-\pi) \, \sqrt{X} \, R_Z(-\pi) \right] \, \text{ECR} \, \left[ X \otimes I \right],
\]

where \( R_Z(\theta) \), \( \sqrt{X} \), and \( X \) are single-qubit operations. This relation neglects a global phase of $\pi/2$. Using this equivalence, the \( R_{zzz\ldots z} \) gate can be decomposed into a sequence of ECR gates and arbitrary single-qubit rotations. The decomposition is illustrated in Figure \ref{fig:RZZZ_ECR}.

\begin{figure}[h!]
    \centering
    \includegraphics[width=0.83\linewidth]{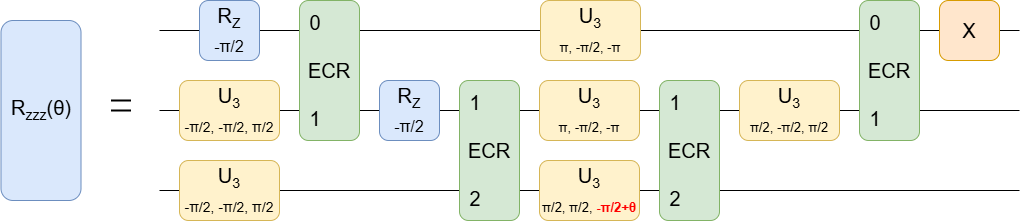}
    \caption{Decomposition of \( R_{zzz\ldots z} \) gate into ECR gates and arbitrary single-qubit rotations; correct up to a global phase.  Here, \( U3 \) gate is the universal single-qubit rotation gate parameterized by three angles \( \theta \), \( \phi \), and \( \lambda \), and its matrix form is given by Eq.\ref{eq:U3} below}
    \label{fig:RZZZ_ECR}
\end{figure}

\begin{equation}
U3(\theta, \phi, \lambda) = \begin{pmatrix} \cos\frac{\theta}{2} & -e^{i\lambda} \sin\frac{\theta}{2} \\
e^{i\phi} \sin\frac{\theta}{2} & e^{i(\phi + \lambda)} \cos\frac{\theta}{2}
\end{pmatrix} \label{eq:U3}
\end{equation}

Implementing the k-qubit \( R_{zzz\ldots z} \) gates on quantum hardware requires only \( 2k \) two-qubit entangling gates. This linear scaling in the number of entangling gates ensures that the circuits remain efficient and practical for hardware implementation, even as \( k \) increases. Additionally, the use of hardware-native gate sets, such as ECR gates, further enhances the performance by minimizing gate overhead and reducing error rates in architectures where these gates are natively supported.

\end{document}